\documentclass[pdflatex,sn-mathphys-num]{sn-jnl}
\usepackage{lipsum}
\usepackage[utf8]{inputenc}
\usepackage{amsmath,amssymb,amsthm,mathtools,bm}
\usepackage{microtype}
\usepackage[utf8]{inputenc}

\usepackage{tikz}
\usepackage{color}
\usepackage{graphicx}
\usepackage{dsfont}
\usepackage{mathrsfs}
\usepackage{verbatim}
\usepackage{enumerate}
\usepackage{tcolorbox}
\usepackage{enumitem}
\tcbuselibrary{breakable}
\usepackage{lipsum}

\theoremstyle{plain}

\theoremstyle{remark}

\usepackage{graphicx}
\usepackage{epic}\usepackage{eepic}\usepackage{epsfig}\usepackage{amsfonts}
\usepackage{latexsym}\usepackage{color}\usepackage{enumerate}\usepackage{bbm}
\allowdisplaybreaks[4]

\hypersetup{
colorlinks,
linkcolor=blue,
anchorcolor=blue,
citecolor=red
}

\theoremstyle{plain}
\newtheorem{theorem}{Theorem}\newtheorem{lemma}[theorem]{Lemma}\newtheorem{proposition}[theorem]{Proposition}\theoremstyle{definition}
\theoremstyle{remark}
\newtheorem{remark}{Remark}\newcommand{\supp}{\operatorname{supp}}

\newcommand{\nb}{\nonumber}

\newcommand{\Tr}{\operatorname{Tr}}

\providecommand{\claimname}{Claim}
\providecommand{\theoremname}{Theorem}

\makeatother
\providecommand{\claimname}{Claim}
\providecommand{\theoremname}{Theorem}

\begin{document}
\title{Strong Converse Exponents of Quantum Soft Covering and Privacy Amplification}
\author{\fnm{Shi-Bing} \sur{Li}\textsuperscript{*}}
\author{\fnm{Hongsen} \sur{Qiu}\textsuperscript{$\dagger$}}
\author{\fnm{Xinyu} \sur{Zhang}\textsuperscript{$\ddagger$}}
\affil[]{%
\makebox[0pt][c]{%
\begin{minipage}{\textwidth}
\centering
Institute for Advanced Study in Mathematics,\\
Harbin Institute of Technology, Harbin, 150001, China
\end{minipage}%
}%
}

\abstract{
We determine the exact strong converse exponent of quantum soft covering
under the sandwiched R{\'e}nyi divergence for all orders
$\alpha\in[\frac{1}{2},\infty)$. For
$\alpha\in[\frac{1}{2},1)$, the exponent is characterized by the
two-parameter club-sandwiched mutual information, whereas for
$\alpha\in[1,\infty)$, it is characterized by the order-$\alpha$
sandwiched R{\'e}nyi mutual information. We also determine the exact
strong converse exponent of privacy amplification against quantum side
information under the sandwiched R{\'e}nyi divergence for
$\alpha\in(2,\infty)$, expressed in terms of the corresponding
order-$\alpha$ sandwiched R{\'e}nyi conditional entropy. To the best of
our knowledge, these results provide the first exact characterization of
the strong converse exponent of quantum soft covering and the first precise operational interpretation of the
two-parameter club-sandwiched mutual information in the quantum setting. The key ingredient is that we establish the exponential rate of the $K$-functional, which is instrumental in deriving the strong converse exponent of quantum soft covering for $\alpha\in[\frac{1}{2},1)$.
}

\maketitle
\begingroup
\renewcommand{\thefootnote}{\fnsymbol{footnote}}

\footnotetext[1]{\texttt{shibingli10@gmail.com}}
\footnotetext[2]{\texttt{chieughongsen@gmail.com}}
\footnotetext[3]{\texttt{xy.zhang@stu.hit.edu.cn}}

\endgroup
\setcounter{footnote}{0}

\section{Introduction}

Consider a classical--quantum (C--Q) channel
$
\mathcal{N}_{X\to E}:x\mapsto\rho_E^x
$
and an input distribution $P_X$ on a finite alphabet $\mathcal X$.  The
corresponding C--Q state and average channel output are
\begin{equation}
\rho_{XE}
=
\sum_{x\in\mathcal X}P_X(x)|x\rangle\langle x|\otimes\rho_E^x,
\qquad
\rho_E
=
\sum_{x\in\mathcal X}P_X(x)\rho_E^x.
\label{eq:intro-cq-state}
\end{equation}
Two fundamental randomization tasks associated with
\eqref{eq:intro-cq-state} are quantum soft covering and quantum privacy
amplification.  Although the two tasks have different operational objectives,
both seek to erase classical--quantum correlations by averaging a randomly
selected family of the states $\{\rho_E^x\}_{x\in\mathcal X}$.

In quantum soft covering, one draws a random codebook
$
\mathcal C=\{X_1,\ldots,X_M\}
$
whose codewords are independently distributed according to $P_X$, and uses
the codebook-induced state
\begin{equation}
\rho_E^{\mathcal C}
:=
\frac{1}{M}\sum_{m=1}^{M}\rho_E^{X_m}
\label{eq:intro-codebook-state}
\end{equation}
to approximate the target state $\rho_E$.  Equivalently, a uniformly chosen
codeword is sent through the C--Q channel and the resulting average output is
required to be nearly indistinguishable from the output generated by the full
input distribution.  Soft covering originated in the classical study of
common information and channel resolvability and has become an important tool
in channel synthesis, wiretap coding, secrecy analysis, and lossy source
coding~\cite{Wyner1975common,HanVerdu1993approximation,Hayashi2006resolvability,
Cuff2013distributed,PTM2016exact,Yaglicuff2019exact,YuTan2018renyi,YuTan2018wyner,Yu2024renyi,YuTan2019exact}.  Its quantum counterpart appears in quantum
identification, measurement compression, channel simulation, and quantum
cryptography~\cite{AhlswedeWinter2002,BennettEtAl2014,DevetakWinter2003,
DevetakWinter2005,OYB2024exponents}.

Privacy amplification starts from the same state $\rho_{XE}$, but has the
opposite operational direction.  Here $X$ is held by a legitimate party and
$E$ represents the quantum side information available to an adversary.  A
hash function $f:\mathcal X\to\mathcal Z$ produces the key $Z=f(X)$ and
induces the state
\begin{equation}
\mathcal R_f(\rho_{XE})
=
\sum_{z\in\mathcal Z}|z\rangle\langle z|
\otimes
\sum_{x:f(x)=z}P_X(x)\rho_E^x.
\label{eq:intro-pa-state}
\end{equation}
The objective is to make this state close to
$
\frac{\mathbbm 1_{\mathcal Z}}{|\mathcal Z|}\otimes\rho_E
$, so that the extracted key is both uniform and independent of the
adversary.  Privacy amplification is a basic primitive in secret-key
distillation and quantum cryptography~\cite{RennerKonig2005,Renner2005security,
TomamichelEtAl2011,BBCM1995generalized,RennerWolf2005simple,DevetakWinter2005}.

In this paper, we study the strong converse exponents of these two quantum information-processing tasks, using the sandwiched Rényi divergence as the distinguishability measure.  This divergence interpolates between several central quantum
information measures: its order-one limit is the quantum relative entropy,
while at order $\frac{1}{2}$ it is, up to the convention for squaring the
fidelity, the negative logarithm of the fidelity.  Most importantly, it obeys
the data-processing inequality under quantum channels exactly in the
operationally relevant range $\alpha\geq\frac{1}{2}$
\cite{MDSFT2013on,WWY2014strong,FrankLieb2013,
AudenaertDatta2015alphaZ,Zhang2020trace}.  We consequently restrict attention
to $\alpha\geq\frac{1}{2}$.  This restriction is structural rather than merely
technical: for $0<\alpha<\frac{1}{2}$, the sandwiched R\'{e}nyi divergence fails
the data-processing inequality in general and therefore is not monotone under
quantum operations.  In that range it loses a basic property required of an
operational distinguishability measure.

\paragraph{Main results.}
Our first contribution is an exact characterization of the strong converse
exponent of quantum soft covering with i.i.d. random codebooks under the
sandwiched R\'{e}nyi divergence for every order
$\alpha\in[\frac{1}{2},\infty)$.  To the best of our knowledge, this is the first exact
strong converse exponent for quantum soft covering.  For orders at least one,
the answer takes the single-order positive-part form below. We write
$|t|^+:=\max\{t,0\}$.

\begin{theorem}
Let
$
\rho_{XE}
=
\sum_{x\in\mathcal X}P_X(x)|x\rangle\langle x|\otimes\rho_E^x
$
be a C--Q state and $R\geq0$.  For every
$\alpha\in[1,\infty)$, we have
\begin{equation}
\varGamma_{\rm sc}^{(\alpha)}(\rho_{XE},R)
=
\left|I_\alpha(X:E)_{\rho_{XE}}-R\right|^+,
\label{eq:intro-sc-alpha-large}
\end{equation}
where the sandwiched R{\'e}nyi mutual information $I_\alpha(X:E)_{\rho_{XE}}$ is defined as \eqref{def:renyi-mutual}.
\end{theorem}

For orders below one, a single R\'{e}nyi order no longer suffices.  The exact
exponent involves an optimization over an auxiliary order and is characterized
by the two-parameter club-sandwiched mutual information.

\begin{theorem}
Let
$
\rho_{XE}
=
\sum_{x\in\mathcal X}P_X(x)|x\rangle\langle x|\otimes\rho_E^x
$
be a C--Q state and $R\geq0$.  For every
$\alpha\in[\frac{1}{2},1)$, we have
\begin{equation}
\varGamma_{\rm sc}^{(\alpha)}(\rho_{XE},R)
=
\sup_{\alpha\leq a\leq1}
\frac{\alpha(1-a)}{a(1-\alpha)}
\left\{
\widetilde{I}_a^{\frac{\alpha-a}{\alpha(1-a)}}
(X:E)_{\rho_{XE}}-R
\right\},
\label{eq:intro-sc-alpha-small}
\end{equation}
where the two-parameter club-sandwiched mutual information $\widetilde{I}_a^{\frac{\alpha-a}{\alpha(1-a)}}
(X:E)_{\rho_{XE}}$ is defined as \eqref{def:club-mutual}.
\end{theorem}

Our second contribution closes the remaining high-order regime for quantum
privacy amplification.  We prove the following exact strong converse exponent
for every $\alpha>2$.

\begin{theorem}\label{thm:5-14}
Let
$
\rho_{XE}
=
\sum_{x\in\mathcal X}P_X(x)|x\rangle\langle x|\otimes\rho_E^x
$
be a C--Q state and
$R\geq 0$.  For every $\alpha\in(2,\infty)$, we have
\begin{equation}
\varGamma_{\rm pa}^{(\alpha)}(\rho_{XE},R)
=
\left|R-H_\alpha(X|E)_{\rho_{XE}}\right|^+,
\end{equation}
where the sandwiched R{\'e}nyi conditional entropy $H_\alpha(X|E)_{\rho_{XE}}$ is defined as \eqref{def:renyi-conditional}.
\end{theorem}

Theorem~\ref{thm:5-14} complements the result of Li, Yao, and Hayashi for
$\alpha\in[1,2]$~\cite{LYH2023tight}.  For the low-order regime, Rubboli and
Tomamichel characterized the endpoint $\alpha=\frac{1}{2}$ under the composable
fidelity criterion~\cite{rubboliTomamichel2026composable}; their argument
extends directly to every $\alpha\in(\frac{1}{2},1)$.  Alternatively, the method
developed here for quantum soft covering below order one can be adapted to
privacy amplification and yields the corresponding club-sandwiched
conditional-entropy formula.  Consequently, combining the present results
with~\cite{LYH2023tight,rubboliTomamichel2026composable} gives exact strong
converse exponents of quantum privacy amplification throughout the entire
data-processing range $
\alpha\in\left[\frac{1}{2},\infty\right).$

\paragraph{Related work.} The exponential behavior of classical soft covering has been studied under
several distinguishability criteria, including total variation and R\'{e}nyi
divergences~\cite{Hayashi2006resolvability,YuTan2018renyi,Yaglicuff2019exact,Yassaee2019almost}.  In the quantum setting, Cheng and Gao
derived one-shot exponential achievability and strong converse bounds for
i.i.d. and constant-composition random codebooks under the trace
distance~\cite{ChengGao2024error}.  The optimal second-order asymptotics under
the same criterion were subsequently characterized in
\cite{ShenGaoCheng2024optimal}.  These works establish the quantum mutual
information as the threshold for trace-distance soft covering, but do not give
an exact strong converse exponent.  Thus, even for i.i.d. random codebooks, an
exact quantum characterization has remained open.

Quantum privacy amplification has an extensive literature.  Its optimal
first-order extraction rate under standard composable security criteria is
the conditional von Neumann entropy~\cite{RennerKonig2005,
Renner2005security}.  Error exponents and strong converse bounds have been
investigated under the trace distance, relative entropy, fidelity, purified
distance, and R\'{e}nyi criteria~\cite{Hayashi2015security,LYH2023tight,
LiYao2024operational,rubboliTomamichel2026composable}.  In particular, Li,
Yao, and Hayashi established the strong converse exponent relevant here for
orders $\alpha\in[1,2]$~\cite{LYH2023tight}, while Rubboli and Tomamichel
obtained the exact composable fidelity exponent, corresponding to the endpoint
$\alpha=\frac{1}{2}$, in terms of a club-sandwiched conditional
entropy~\cite{rubboliTomamichel2026composable}.

\section{Preliminaries}
\subsection{Basic Notations}
Let $\mathcal{H}$ be a Hilbert space associated with a finite-dimensional
quantum system. We denote the set of bounded linear operators on $\mathcal{H}$
by $\mathcal{B(H)}$, and the positive semi-definite operators by
$\mathcal{P(H)}$. We use the notations $\mathcal{D(H)}$ to represent
the set of quantum states. We use $\mathbbm{1}_{\mathcal{H}}$ to represent the identity
operator on $\mathcal{H}$. When $\mathcal{H}$ is associated
with a system $A$, the above notations $\mathcal{B(H)},\mathcal{P(H)},\mathcal{D(H)}$
and $\mathbbm{1}_{\mathcal{H}}$ are also written as $\mathcal{B}(A),\mathcal{P}(A),\mathcal{D}(A)$
and $\mathbbm{1}_{A}$, respectively. The dimension of system $A$
is denoted by $|A|$. A classical-quantum (C-Q) state is a bipartite
state of the form
\begin{equation*}
\rho_{XE}=\sum_{x\in\mathcal{X}}q(x)|x\rangle\langle x|_{X}\otimes\rho_{E}^{x},
\end{equation*}
where $\rho_{E}^{x}\in\mathcal{D}(E)$, $\{|x\rangle\}$ is an orthonormal basis of $\mathcal{H}_{X}$ and $q(x)$ is a probability
distribution on a finite set $\mathcal{X}$. In this sense, we can view $\rho_{E}^{X}$ as a discrete matrix-valued random variable with probability
distribution $P(X=x)=q(x)$ and it is finitely supported.

Let $H$ be a self-adjoint operator on $\mathcal{H}$ with spectral
projections $P_{1},...,P_{q}$. Then the pinching map associated with
$H$ is defined as
\begin{equation*}
\mathcal{E}_{H}:X\mapsto\sum_{i=1}^{q}P_{i}XP_{i}.
\end{equation*}
The pinching inequality~\cite{Hayashi2002optimal} states that for any $\sigma\in\mathcal{P(H)}$,
we have
\begin{equation*}
\sigma\leq \nu(H)\mathcal{E}_{H}(\sigma),
\end{equation*}
where $\nu(H)$ is the number of distinct eigenvalues of $H$.

Throughout this paper, the functions ``log'' and ``exp'' are with base $2$. For any $\alpha\in(0,1)\cup(1,\infty)$, the order-$\alpha$ fidelity between a quantum state $\rho\in\mathcal{D}(\mathcal{H})$ and a semi-definite operator $\sigma\in\mathcal{P}(\mathcal{H})$ is defined as
\begin{equation*}
Q_{\alpha}(\rho\|\sigma)
:=\mathrm{Tr} \left(\sigma^{\frac{1-\alpha}{2\alpha}}\rho\,\sigma^{\frac{1-\alpha}{2\alpha}}\right)^{\alpha}.
\end{equation*}
For any $\alpha\in(0,1)\cup(1,\infty)$, the corresponding order-$\alpha$ sandwiched R{\'e}nyi divergence
~\cite{MDSFT2013on,WWY2014strong} is defined as
\begin{equation*}
D_\alpha(\rho\|\sigma)
:=\frac{1}{\alpha-1}\log Q_\alpha(\rho\|\sigma).
\end{equation*}
Taking the limit $\alpha\to1$, it reduces to the quantum relative entropy~\cite{Umegaki1954conditional}
\begin{equation*}
D_1(\rho\|\sigma):=\lim_{\alpha\to1} D_\alpha(\rho\|\sigma)
= D(\rho\|\sigma)
:=\mathrm{Tr}\,\rho(\log \rho - \log \sigma).
\end{equation*}

Let $\rho_{XE}=\sum_{x\in\mathcal{X}} p(x)\,|x\rangle\langle x|_X \otimes \rho_E^x$ be a C-Q state. For any $\alpha\in(0,\infty)$, the order-$\alpha$ sandwiched Rényi conditional entropy~\cite{TBH2014relating} is defined as
\begin{align}
H_\alpha(X|E)_{\rho}
&:= - D_\alpha(\rho_{XE}\,\|\,\mathbbm{1}_{\mathcal{X}} \otimes \rho_E).\label{def:renyi-conditional}
\end{align}
Similarly, the order-$\alpha$ sandwiched Rényi mutual information is defined as
\begin{align}
I_\alpha(X:E)_{\rho}
&:= D_\alpha(\rho_{XE}\,\|\,\rho_X \otimes \rho_E).\label{def:renyi-mutual}
\end{align}
In the limit $\alpha \to 1$, these quantities recover the standard von Neumann conditional entropy and mutual information, i.e.,
\begin{align*}
H(X|E)_{\rho_{XE}}
&:= H(\rho_{XE}) - H(\rho_E), \\
I(X:E)_{\rho_{XE}}
&:= D(\rho_{XE}\,\|\,\rho_X \otimes \rho_E)
= H(\rho_X) + H(\rho_E) - H(\rho_{XE}),
\end{align*}
where $H(\rho) := -\mathrm{Tr}\,\rho \log \rho$ is the von Neumann entropy.

For any $\alpha\in(0,\infty)$, the
mixed-order order-two R{\'e}nyi conditional entropy and R{\'e}nyi mutual information are defined as \cite{liQiuZhang2026reliability}
\begin{align}
H_{2}^{(\alpha)}(X|E)_{\rho_{XE}}:=&
-\log
\sum_{x\in\mathcal X}P_X(x)^2
\operatorname{Tr}
\left[
\left(
\rho_E^{\frac{1-\alpha}{2\alpha}}
\rho_E^x
\rho_E^{\frac{1-\alpha}{2\alpha}}
\right)^2
\rho_E^{\frac{\alpha-2}{\alpha}}
\right], \\
I_{2}^{(\alpha)}(X:E)_{\rho_{XE}}
:=&
\log
\sum_{x\in\mathcal X}P_X(x)
\operatorname{Tr}
\left[
\left(
\rho_E^{\frac{1-\alpha}{2\alpha}}
\rho_E^x
\rho_E^{\frac{1-\alpha}{2\alpha}}
\right)^2
\rho_E^{\frac{\alpha-2}{\alpha}}
\right].
\end{align}

In the next proposition, we collect a few properties of the R{\'e}nyi quantities defined above.
\begin{lemma}\label{lem:proverty}
Let $\rho,\sigma \in \mathcal{D}(\mathcal{H})$. Then the sandwiched Rényi divergence satisfies the following properties.
\begin{enumerate}
\item Monotonicity~\cite{MDSFT2013on,Beigi2013sandwiched}:
For any $0 < \alpha < \beta$, we have
\begin{equation*}
D_\alpha(\rho\|\sigma) \le D_\beta(\rho\|\sigma).
\end{equation*}
\item Additivity~\cite{MullerLennertEtAl2013}:
For any $\rho_1,\rho_2,\sigma_1,\sigma_2$,
\begin{equation*}
D_\alpha(\rho_1\otimes\rho_2 \,\|\, \sigma_1\otimes\sigma_2)
= D_\alpha(\rho_1\|\sigma_1) + D_\alpha(\rho_2\|\sigma_2).
\end{equation*}
\end{enumerate}
\end{lemma}

For later use, we fix our convention for Schatten norms.  For
$A\in\mathcal B(\mathcal H)$, let $|A|=(A^*A)^{1/2}$.  If
$1\leq p<\infty$, set
\[
    \|A\|_{S_p(\mathcal H)}
    :=
    \bigl(\operatorname{Tr}|A|^p\bigr)^{1/p},
\]
and let
$\|A\|_{S_\infty(\mathcal H)}:=\|A\|_{\mathcal B(\mathcal H)}$.
Since $\mathcal H$ is finite-dimensional, every operator belongs to every
Schatten class.  If $1/p+1/p'=1$, the trace pairing gives
\[
    \bigl|\operatorname{Tr}(A^*B)\bigr|
    \leq
    \|A\|_{S_p(\mathcal H)}\|B\|_{S_{p'}(\mathcal H)},
    \qquad
    \|A\|_{S_p(\mathcal H)}
    =
    \sup_{\|B\|_{S_{p'}(\mathcal H)}\leq1}
    \bigl|\operatorname{Tr}(A^*B)\bigr|.
\]
More generally, if $1/r=1/p+1/q$, then
\[
    \|AB\|_{S_r(\mathcal H)}
    \leq
    \|A\|_{S_p(\mathcal H)}\|B\|_{S_q(\mathcal H)},
\]
and the tensor-product norm is multiplicative:
\[
    \|A\otimes B\|_{S_p(\mathcal H\otimes\mathcal K)}
    =
    \|A\|_{S_p(\mathcal H)}\|B\|_{S_p(\mathcal K)}.
\]
We use these facts in their standard forms; see
\cite[Chapter~2]{Simon2005}. We will simply write $\|\cdot\|_p=\|\cdot\|_{S_p(\mathcal H)}$ without ambiguity. 

\subsection{Two-Parameter club-sandwiched mutual information}
In this paper, we introduce a new R{\'e}nyi information quantity,
termed the two-parameter club-sandwiched mutual information. For any
C-Q state $\rho_{XE}$ and parameters $\alpha<1$ and $\lambda<0$, it is
defined as
\begin{align}
\widetilde{I}_\alpha^{\lambda}(X:E)_{\rho_{XE}}
:={}&
\sup_{\sigma_E\in\mathcal{D}(E)}
\frac{1}{\alpha-1}
\log\operatorname{Tr}
\Big[
(\rho_X\otimes\sigma_E)^{
\frac{\lambda}{2}\frac{1-\alpha}{\alpha}}
(\rho_X\otimes\rho_E)^{
\frac{1-\lambda}{2}\frac{1-\alpha}{\alpha}}
\notag\\
&\qquad\quad{}\cdot
\rho_{XE}
(\rho_X\otimes\rho_E)^{
\frac{1-\lambda}{2}\frac{1-\alpha}{\alpha}}
(\rho_X\otimes\sigma_E)^{
\frac{\lambda}{2}\frac{1-\alpha}{\alpha}}
\Big]^\alpha .
\label{def:club-mutual}
\end{align}
From the definition, it is clear that at the point where the maximum is attained, the optimizer $\sigma_E$ satisfies $\mathrm{supp}(\rho_E) \subseteq \mathrm{supp}(\sigma_E)$.

For fixed $\alpha\in[\frac12,1)$, define
\[
\lambda_{\alpha,a}
:=
\frac{\alpha-a}{\alpha(1-a)},
\qquad a\in(\alpha,1).
\]
Following the boundary convention used in
\cite[Lemma~17]{rubboliTomamichel2026composable}, the values at
$a=\alpha$ and $a=1$ are understood through one-sided limits. More
precisely,
\begin{align*}
\lim_{a\downarrow\alpha}
\widetilde I_a^{\lambda_{\alpha,a}}(X:E)_{\rho_{XE}}
&=
I_\alpha(X:E)_{\rho_{XE}},\\
\lim_{a\uparrow1}
\widetilde I_a^{\lambda_{\alpha,a}}(X:E)_{\rho_{XE}}
&=I(X:E)_{\rho_{XE}}.
\end{align*}
Consequently, the objective function in
\eqref{eq:intro-sc-alpha-small} admits the boundary values
\begin{align*}
&\left.
\frac{\alpha(1-a)}{a(1-\alpha)}
\left\{
\widetilde I_a^{\lambda_{\alpha,a}}(X:E)_{\rho_{XE}}-R
\right\}
\right|_{a=\alpha}
=I_\alpha(X:E)_{\rho_{XE}}-R,
\\
&\left.
\frac{\alpha(1-a)}{a(1-\alpha)}
\left\{
\widetilde I_a^{\lambda_{\alpha,a}}(X:E)_{\rho_{XE}}-R
\right\}
\right|_{a=1}
=0.
\end{align*}
Thus, the supremum in \eqref{eq:intro-sc-alpha-small} is well defined
over the closed interval $[\alpha,1]$. In particular, the endpoint
$a=1$ guarantees that the right-hand side is nonnegative.

% We shall also use a weighted version, always under an explicit commutation
% hypothesis.  Let $\omega\in\mathcal D(\mathcal H)$, compress all operators
% to $\operatorname{supp}\omega$, and suppose that $[A,\omega]=0$.  On the
% finite von Neumann algebra $\{\omega\}'$, the functional
% \[
%     \tau_\omega(X):=\operatorname{Tr}(\omega X)
% \]
% is a faithful normal tracial state.  Indeed, for $X,Y\in\{\omega\}'$,
% \[
%     \tau_\omega(XY)
%     =\operatorname{Tr}(\omega XY)
%     =\operatorname{Tr}(Y\omega X)
%     =\operatorname{Tr}(\omega YX)
%     =\tau_\omega(YX).
% \]
% Consequently,
% \[
%     \|A\|_{p,\omega}
%     :=
%     \bigl(\operatorname{Tr}(\omega|A|^p)\bigr)^{1/p}
% \]
% is precisely the noncommutative $L_p$-norm associated with
% $(\{\omega\}',\tau_\omega)$.  If $A_\xi$ is a discrete random
% operator satisfying $[A_\xi,\omega]=0$ almost surely, then
% \[
%     \bigl(\mathbb E_\xi\operatorname{Tr}
%         (\omega|A_\xi|^p)\bigr)^{1/p}
% \]
% is the noncommutative $L_p$-norm in
% $L_\infty(\Omega)\,\overline\otimes\,\{\omega\}'$, equipped with the trace
% $\mathbb E_\xi\otimes\tau_\omega$.  This is the precise meaning of the
% term ``weighted Schatten norm'' used below; see also
% \cite{PisierXu2003} for the semifinite(tracial) noncommutative $L_p$ framework.

\section{Problem Statement and Main Results}
\subsection{Soft covering and main results}
Let $\mathcal N:\mathcal X\to\mathcal D(\mathcal H_E)$ be a
classical--quantum channel given by $
x\longmapsto\rho_E^x$
and $P_X$ be an input distribution on $\mathcal X$. The
corresponding average output state is
\[
\rho_E
:=
\sum_{x\in\mathcal X}P_X(x)\rho_E^x.
\]
The goal of quantum soft covering is to approximate the average output
state $\rho_E$ using the uniform mixture of the channel output states
associated with a finite random codebook. More precisely, let $
\mathcal C=\{X(1),\ldots,X(M)\}$
be a random codebook whose codewords are independently drawn according
to $P_X$. The state induced by the codebook $\mathcal C$ is
\[
\rho_E^{\mathcal C}
:=
\frac{1}{M}\sum_{m=1}^M\rho_E^{X(m)}.
\]
The quantum soft-covering problem asks how large the codebook must be
so that $\rho_E^{\mathcal C}$ closely approximates $\rho_E$.

In this work, the approximation error is measured by the order-$\alpha$
sandwiched R{\'e}nyi divergence. Introducing the classical codebook
register
\[
\rho_{\mathcal CE}
:=
\sum_cP_{\mathcal C}(c)
|c\rangle\langle c|\otimes\rho_E^c,
\qquad
\rho_{\mathcal C}
:=
\sum_cP_{\mathcal C}(c)|c\rangle\langle c|,
\]
the approximation error is defined as
\begin{equation*}
D_\alpha
\left(
\rho_{\mathcal CE}
\,\middle\|\,
\rho_{\mathcal C}\otimes\rho_E
\right)=
\frac{1}{\alpha-1}
\log
\mathbb E_{\mathcal C}
Q_\alpha
\left(
\rho_E^{\mathcal C}
\,\middle\|\,
\rho_E
\right).
\end{equation*}
For the memoryless extension, let
\[
\mathcal C_n
=
\{X^n(1),\ldots,X^n(2^{nR})\},\footnote{\label{fn:common} Without loss of generality, we assume that $2^R$ is an
	integer. Consequently, $2^{nR}$ is an integer for every
	$n\in\mathbb{N}$. This convention has no effect on the strong converse exponent.}
\]
where the positive number $R$ is called code rate and each codeword is independently drawn according to
$P_X^{\otimes n}$. The induced output state is
\[
\rho_{E^n}^{\mathcal C_n}
=
\frac{1}{2^{nR}}
\sum_{m=1}^{2^{nR}}
\rho_{E^n}^{X^n(m)},
\]
where $
\rho_{E^n}^{x^n}
=
\bigotimes_{i=1}^n\rho_E^{x_i}.$ Then, the strong converse exponent of soft covering is defined as
\begin{equation*}
\varGamma_{\rm sc}^{(\alpha)}(\rho_{XE},R):=\liminf_{n\to\infty}\frac{1}{n}D_\alpha
\left(
\rho_{\mathcal{C}_nE^n}
\,\middle\|\,
\rho_{\mathcal C_n}\otimes\rho_E^{\otimes n}
\right).
\end{equation*}

As announced in the Introduction, our first main result determines the
exact strong converse exponent of quantum soft covering for i.i.d. random
codebooks under the sandwiched R{\'e}nyi divergence of every order
$\alpha\in[\frac{1}{2},\infty)$. To the best of our knowledge, this is
the first exact characterization of a strong converse exponent in the
quantum soft-covering setting. For convenience, we restate the result
below, separately for the two regimes $\alpha\in[1,\infty)$ and
$\alpha\in[\frac{1}{2},1)$.

\begin{theorem}\label{thm:sc-sc}
Let $\rho_{XE}=\sum_{x\in \mathcal{X}}P_{X}(x)|x\rangle\langle x|\otimes\rho_{E}^{x}$
be a C-Q state and rate $R\geq0$. For any $\alpha \in [1,\infty)$, we have
\begin{equation*}
\varGamma_{\rm sc}^{(\alpha)}(\rho_{XE},R)=|I_{\alpha}(X:E)_{\rho_{XE}}-R|^{+}.
\end{equation*}
\end{theorem}
\begin{theorem}\label{thm:sc-sc-1/2}
Let $\rho_{XE}=\sum_{x\in \mathcal{X}}P_{X}(x)|x\rangle\langle x|\otimes\rho_{E}^{x}$
be a C-Q state and rate $R\geq0$. For any $\alpha \in [\frac{1}{2},1)$, we have
\begin{equation*}
\varGamma_{\rm sc}^{(\alpha)}(\rho_{XE},R)=\sup_{\alpha\leq a\leq1}\frac{\alpha\left(1-a\right)}{a\left(1-\alpha\right)}\left\{\widetilde{I}_a^{\frac{\alpha-a}{\alpha(1-a)}}(X:E)_{\rho_{XE}}-R\right\}.
\end{equation*}
\end{theorem}

\subsection{Privacy amplification and main results}
Let $\rho_{XE}=\sum_{x\in\mathcal{X}} P_X(x)\, |x\rangle\langle x| \otimes \rho_E^x$
be a C-Q state.
The goal of quantum privacy amplification is to extract from $X$ a key that is approximately uniform and independent of the quantum system $E$. To this end, let $f:\mathcal{X}\to\mathcal{Z}$ be a hash function.
The C-Q state induced by hash function $f$ is
\begin{equation*}
\mathcal{R}_{f}(\rho_{XE})
=\sum_{x\in\mathcal{X}} P_X(x)\, |f(x)\rangle\langle f(x)|_Z \otimes \rho_E^x.
\end{equation*}
The ideal target state is $
\frac{\mathbbm{1}_{\mathcal{Z}}}{|\mathcal{Z}|}\otimes \rho_E$ and $
\rho_E=\sum_{x\in\mathcal{X}}P_X(x)\rho_E^x.$ To quantify the performance of quantum privacy amplification, we measure the discrepancy via the sandwiched Rényi divergence of order-$\alpha$. That is
\begin{equation*}
D_\alpha\left(\mathcal{R}_{f}(\rho_{XE})\Big\|\frac{\mathbbm{1}_{\mathcal{Z}}}{|\mathcal{Z}|}\otimes \rho_E\right).
\end{equation*}
We consider the n-shot extension of quantum privacy amplification. Let
\begin{equation*}
f_n:\mathcal{X}^n \to \mathcal{Z}_n=\{1,2,\cdots,2^{nR}\},
\end{equation*}
where $R\ge 0$ is the key rate. The induced state is
\begin{equation*}
\mathcal{R}_{f_n}(\rho_{XE}^{\otimes n})
=\sum_{x^n\in\mathcal{X}^n} P_X^{\otimes n}(x^n)\,
|f_n(x^n)\rangle\langle f_n(x^n)| \otimes \rho_{E^n}^{x^n},
\end{equation*}
where $\rho_{E^n}^{x^n}=\bigotimes_{i=1}^n \rho_E^{x_i}$.
The performance in the $n$-shot setting is characterized by
\begin{equation*}
D_\alpha\left(\mathcal{R}_{f_n}(\rho_{XE}^{\otimes n})\Big\|\frac{\mathbbm{1}_{\mathcal{Z}_n}}{|\mathcal{Z}_n|}\otimes \rho_E^{\otimes n}\right).
\end{equation*}
The strong converse exponent of quantum privacy amplification is defined as
\begin{align}\label{def:str-pa}
\varGamma_{\rm pa}^{(\alpha)}(\rho_{XE},R)
:=&\liminf_{n\to\infty}
\frac{1}{n} \min_{f_n:\mathcal{X}^{n}\to \mathcal{Z}_n}D_\alpha\left(\mathcal{R}_{f_n}(\rho_{XE}^{\otimes n})\Big\|\frac{\mathbbm{1}_{\mathcal{Z}_n}}{|\mathcal{Z}_n|}\otimes \rho_E^{\otimes n}\right).
\end{align}

As announced in the Introduction, our second main result determines the
exact strong converse exponent of quantum privacy amplification under
the sandwiched R{\'e}nyi divergence for every order
$\alpha\in(2,\infty)$. For convenience, we restate the result below.
\begin{theorem}\label{thm:5-15}
Let $\alpha\in(2,\infty)$, $\rho_{XE}=\sum_{x\in\mathcal{X}} P_X(x)\, |x\rangle\langle x| \otimes \rho_E^x$
be a C-Q state and $R\geq 0$. We have
\begin{equation}
\varGamma_{\rm pa}^{(\alpha)}(\rho_{XE},R)= |R-H_\alpha(X|E)_{\rho_{XE}}|^+.\label{result:PA}
\end{equation}
\end{theorem}
\begin{remark}
In \cite{LYH2023tight}, Li, Yao and Hayashi proved that identity \eqref{result:PA} holds for all $\alpha\in[1,2]$. Rubboli and Tomamichel \cite{rubboliTomamichel2026composable} derived the strong converse exponent for privacy amplification in terms of the order-$\frac{1}{2}$ fidelity. In fact, their method can be trivially extended to the full range $\alpha\in(\frac{1}{2},1)$. Combining the results from \cite{LYH2023tight}, \cite{rubboliTomamichel2026composable} and Theorem~\ref{thm:5-15}, we thus obtain the strong converse exponents of privacy amplification for all orders $\alpha\in[\frac{1}{2},\infty)$. Furthermore, our proposed approach for soft covering with $\alpha\in(\frac{1}{2},1)$ can also be adopted to derive the strong converse exponent of privacy amplification for $\alpha\in(\frac{1}{2},1)$.
\end{remark}

\section{Strong Converse Exponent of Soft Covering for $\alpha\in[1/2,1)$}
In this section, we prove Theorem~\ref{thm:sc-sc-1/2}, which
characterizes the strong converse exponent of quantum soft covering for
R{\'e}nyi orders $\alpha\in[1/2,1)$. We first establish the exponential rate of the $K$-functional related to the noncommutative $L_p$-space and conditional column $L_p$-space. We then use the result to give an upper bound of the strong converse exponent of quantum soft covering. Finally, we prove the lower bound for such exponent. 
\subsection{The $K$-functional method}
\label{K-func}

In this subsection, we introduce the $K$-functional associated with a
noncommutative $L_p$ space and a conditional column $L_p$ space.

The tracial noncommutative $L_p$-spaces used below are standard; see
\cite{PisierXu2003} and \cite[Section~1.1]{JungeParcet2010mixed}.  The precise
conditional row and column spaces, their completions, and their isometric
dualities are given in \cite[Section~1.3]
{JungeParcet2010mixed}; see also the original construction in
\cite{Junge2002}.  The more general conditional $L_p$
framework is developed in \cite[Chapter~4]{JungeParcet2010mixed}.  For Peetre's
$K$-functional and the interpolation convention used below, see
\cite[Section~3.1]{BerghLofstrom1976interpolation}.

Let $(\mathcal M,\tau_{\mathcal M})$ be a finite von Neumann algebra
equipped with a normal faithful finite trace, let
$\mathcal N\subseteq\mathcal M$ be a von Neumann subalgebra, and let
\[
    \mathcal E:\mathcal M\longrightarrow\mathcal N
\]
be the normal trace-preserving conditional expectation.  We write
$\tau_{\mathcal N}:=\tau_{\mathcal M}|_{\mathcal N}$.  For
$1\leq p<\infty$, the noncommutative $L_p$-space
$L_p(\mathcal M)$ is the completion of $\mathcal M$ with respect to
\[
    \|x\|_{L_p(\mathcal M)}
    :=
    \bigl(\tau_{\mathcal M}(|x|^p)\bigr)^{1/p}.
\]
The conditional column space $L_p^c(\mathcal M,\mathcal E)$ is the
completion of $\mathcal M$ with respect to
\[
    \|x\|_{L_p^c(\mathcal M,\mathcal E)}
    :=
    \bigl\|\mathcal E(x^*x)^{1/2}\bigr\|_{L_p(\mathcal N)}
    =\left(\tau_{\mathcal{N}}\big(\mathcal E(x^*x)^{p/2}\big)\right)^{1/p}.
\]
% When $1\leq p<2$, the last expression is first defined for
% $x\in\mathcal M$, and $L_p^c(\mathcal M,\mathcal E)$ is then obtained
% by completion.  In particular, one must not tacitly regard
% $\mathcal E$ as a continuous map on all of $L_{p/2}(\mathcal M)$,
% because $p/2<1$.

For $1<p<\infty$, let $p'$ be determined by
$1/p+1/p'=1$.  With respect to the anti-linear trace pairing
\[
    \langle x,G\rangle
    :=
    \tau_{\mathcal M}(x^*G),
\]
one has the isometric dualities
\[
    L_p(\mathcal M)^*=L_{p'}(\mathcal M),
    \qquad
    L_p^c(\mathcal M,\mathcal E)^*
    =L_{p'}^c(\mathcal M,\mathcal E).
\]
The second identity is exactly the conditional-column duality in
\cite[Section~1.3]{JungeParcet2010mixed}.

On the $n$-fold tensor product, define
\begin{align*}
D_{p,n}(x)=
\|x\|_{L_p(\mathcal M^{\bar\otimes n})},\quad
C_{p,n}(x)
=
\|x\|_{
L_p^c(\mathcal M^{\bar\otimes n},\mathcal E^{\otimes n})
}.
\end{align*}
Both norms are multiplicative under tensor products. Namely,
\begin{align*}
D_{p,n+m}(x\otimes y)
=
D_{p,n}(x)D_{p,m}(y),\quad
C_{p,n+m}(x\otimes y)
=
C_{p,n}(x)C_{p,m}(y).
\end{align*}
We also simply denote
$$ D_{p,n}=L_p(\mathcal M^{\bar\otimes n}),\quad C_{p,n}=L_p^c(\mathcal M^{\bar\otimes n},\mathcal E^{\otimes n}).  $$
Both $D_{p,n}$ and $C_{p,n}$ are reflexive.

Next, we turn to the case $1<p\leq2$. Since $p/2\in(0,1]$, operator
concavity and the trace-preserving property of $\mathcal E^{\otimes n}$
give
\begin{align*}
C_{p,n}(x)^p=
\tau_{\mathcal N}^{\otimes n}
\left[
\bigl(\mathcal E^{\otimes n}(x^*x)\bigr)^{p/2}
\right]
\geq
\tau_{\mathcal N}^{\otimes n}
\left[
\mathcal E^{\otimes n}
\bigl((x^*x)^{p/2}\bigr)
\right]
=
\tau_{\mathcal M}^{\otimes n}
\bigl((x^*x)^{p/2}\bigr)
=
D_{p,n}(x)^p.
\end{align*}
Thus
$
D_{p,n}(x)\leq C_{p,n}(x)$ and one has the embedding $C_{p,n}\hookrightarrow D_{p,n}$.

For $t>0$, the corresponding $K$-functional is defined by
\begin{equation*}
K_{p,n,t}(x)
:=
\inf_{x=y+z}
\left\{
tD_{p,n}(y)+C_{p,n}(z)
\right\},\quad x\in D_{p,n}+C_{p,n},
\end{equation*}
where $K_{p,n,t}(\cdot)$ is the norm of the sum space $tD_{p,n}+C_{p,n}$. We simply denote $K_{p,n,t}=tD_{p,n}+C_{p,n}$.
\begin{remark}
With the standard Peetre convention
\[
    K(t,x;X_0,X_1)
    :=
    \inf_{x=x_0+x_1}
    \bigl\{\|x_0\|_{X_0}+t\|x_1\|_{X_1}\bigr\},
\]
see \cite[Section~3.1]{BerghLofstrom1976interpolation}, the functional used here is
precisely
\[
    K_{p,n,t}(x)
    =K(t,x;C_{p,n},D_{p,n})
    =t\,K(t^{-1},x;D_{p,n},C_{p,n}).
\]
Thus our convention is the standard Peetre $K$-functional applied to
the reversed compatible couple $(C_{p,n},D_{p,n})$.
\end{remark}

With respect to the duality pairing
\[
\langle x,G\rangle
:=
\tau_{\mathcal M^{\bar{\otimes} n}}(x^*G),
\]
the standard dualities are
\[
D_{p,n}^*
=
D_{p^\prime,n},\quad
C_{p,n}^*
=
C_{p^\prime,n}.
\]
% Consequently, the dual unit ball of the norm
% $K_{p,n,t}$ is given by
% \[
% D_{p',n}(G)\leq t,
% \qquad
% C_{p',n}(G)\leq1.
% \]
If $G$ satisfies $D_{p',n}(G)\leq t$ and $
C_{p',n}(G)\leq1$, then every decomposition
$x=y+z$ satisfies
\[
\begin{aligned}
|\langle x,G\rangle|
&\leq
|\langle y,G\rangle|
+
|\langle z,G\rangle|
\\
&\leq
D_{p,n}(y)D_{p',n}(G)
+
C_{p,n}(z)C_{p',n}(G)
\\
&\leq
tD_{p,n}(y)+C_{p,n}(z).
\end{aligned}
\]
Taking the infimum over all decompositions gives
$$ K_{p,n,t}(x)
\geq
\sup_G
\left\{
|\langle x,G\rangle|:
D_{p',n}(G)\leq t,\
C_{p',n}(G)\leq1
\right\}.$$
 Conversely, if $G\in D_{p^\prime,n}\cap C_{p^\prime,n}$ belongs to the dual unit ball of $K_{p,n,t}$, then for every $x$,
\[
|\langle x,G\rangle|
\leq
K_{p,n,t}(x)
\leq
tD_{p,n}(x),\quad
|\langle x,G\rangle|
\leq
K_{p,n,t}(x)
\leq
C_{p,n}(x),
\]
where we use the decompositions $x=x+0$ and
$x=0+x$. Thus
\[
D_{p',n}(G)\leq t,
\qquad
C_{p',n}(G)\leq1.
\]
Hence one can conclude that
\begin{equation}\label{dual-def}
K_{p,n,t}(x)
=
\sup_G
\left\{
|\langle x,G\rangle|:
D_{p',n}(G)\leq t,\
C_{p',n}(G)\leq1
\right\}.
\end{equation}
Applying the tensoring multiplicativity in
\eqref{dual-def} gives
\begin{equation}
K_{p,n+m,t_1t_2}(x\otimes y)
\geq
K_{p,n,t_1}(x)K_{p,m,t_2}(y),
\qquad t_1,t_2>0.
\label{eq:abstract-K-supermultiplicativity}
\end{equation}

% Next, we turn to the case $1<p\leq2$. Since $p/2\in(0,1]$, operator
% concavity and the trace-preserving property of $\mathcal E^{\otimes n}$
% give
% \begin{align*}
% C_{p,n}(x)^p=
% \tau_{\mathcal N}^{\otimes n}
% \left[
% \bigl(\mathcal E^{\otimes n}(x^*x)\bigr)^{p/2}
% \right]
% \geq
% \tau_{\mathcal N}^{\otimes n}
% \left[
% \mathcal E^{\otimes n}
% \bigl((x^*x)^{p/2}\bigr)
% \right]
% =
% \tau_{\mathcal M}^{\otimes n}
% \bigl((x^*x)^{p/2}\bigr)
% =
% D_{p,n}(x)^p.
% \end{align*}
% Thus
% $
% D_{p,n}(x)\leq C_{p,n}(x)$.
For $0<t\leq1$, every decomposition $x=y+z$ satisfies
\[
\begin{aligned}
tD_{p,n}(y)+C_{p,n}(z)
\geq
tD_{p,n}(y)+D_{p,n}(z)\geq
t\bigl(D_{p,n}(y)+D_{p,n}(z)\bigr)
\geq
tD_{p,n}(x).
\end{aligned}
\]
Taking the infimum one obtains $K_{p,n,t}(x)\geq tD_{p,n}(x)$. Moreover, the decomposition $x=x+0$ gives $K_{p,n,t}(x)\leq tD_{p,n}(x)$. Hence
\begin{equation}\label{small-t}
K_{p,n,t}(x)=tD_{p,n}(x),
\qquad 0<t\leq1.
\end{equation}
For $t\geq1$, a similar argument gives
\begin{equation*}
D_{p,n}(x)
\leq
K_{p,n,t}(x)
\leq
C_{p,n}(x),
\qquad t\geq1.
\end{equation*}

For $0<t_1\leq t_2$, every decomposition $x=y+z$ satisfies
\[
t_1D_{p,n}(y)+C_{p,n}(z)
\leq
t_2D_{p,n}(y)+C_{p,n}(z)
\leq
\frac{t_2}{t_1}
\left(
t_1D_{p,n}(y)+C_{p,n}(z)
\right).
\]
Taking the infimum over all decompositions yields
\begin{equation}
K_{p,n,t_1}(x)
\leq
K_{p,n,t_2}(x)
\leq
\frac{t_2}{t_1}K_{p,n,t_1}(x).
\label{eq:abstract-K-monotonicity}
\end{equation}

Finally, fix a nonzero $x\in K_{p,n,t}$ and define
\begin{equation*}
k_{p,n,r}(x)
:=
\frac{1}{n}
\log
K_{p,n,2^{nr}}(x^{\otimes n}).
\end{equation*}
By \eqref{eq:abstract-K-supermultiplicativity},
\[
\log
K_{p,n+m,2^{(n+m)r}}(x^{\otimes(n+m)})
\geq
\log K_{p,n,2^{nr}}(x^{\otimes n})+
\log K_{p,m,2^{mr}}(x^{\otimes m}).
\]
Fekete's lemma therefore implies that
\begin{equation*}
k_{p,r}(x)
:=
\lim_{n\to\infty}k_{p,n,r}(x)
=
\sup_{n\geq1}k_{p,n,r}(x).
\end{equation*}

Putting $t=2^{nr}$, the estimates on $K_{p,n,t}$ above give
\[
k_{p,n,r}(x)
=
r+\log D_{p,1}(x),
\qquad r\leq0,
\]
and
\[
\log D_{p,1}(x)
\leq
k_{p,n,r}(x)
\leq
\log C_{p,1}(x),
\qquad r\geq0.
\]
Thus $k_{p,n,r}(x)$ and $k_{p,r}(x)$ are finite. Moreover,
\eqref{eq:abstract-K-monotonicity} implies that, for $r_1\leq r_2$,
\begin{equation*}
0
\leq
k_{p,n,r_2}(x)-k_{p,n,r_1}(x)
\leq
r_2-r_1.
\end{equation*}
Hence both $r\mapsto k_{p,n,r}(x)$ and
$r\mapsto k_{p,r}(x)$ are nondecreasing and $1$-Lipschitz.

It remains to record the concavity of $k_{p,r}$. Let $\eta=j/(j+\ell)$, where $j,\ell\geq1$ are integers. Applying
\eqref{eq:abstract-K-supermultiplicativity} to $jN$ copies with
$t=2^{Nr_0}$ and $\ell N$ copies with $t=2^{Nr_1}$, we obtain
\begin{align*}
K_{p,(j+\ell)N,\,
2^{(j+\ell)N(\eta r_0+(1-\eta)r_1)}}
\bigl(x^{\otimes (j+\ell)N}\bigr)
\geq
K_{p,N,2^{Nr_0}}\bigl(x^{\otimes N}\bigr)^j
K_{p,N,2^{Nr_1}}\bigl(x^{\otimes N}\bigr)^\ell .
\end{align*}
Taking logarithms and dividing by $(j+\ell)N$ gives
\[
k_{p,(j+\ell)N,\eta r_0+(1-\eta)r_1}(x)
\geq
\eta k_{p,N,r_0}(x)
+(1-\eta)k_{p,N,r_1}(x).
\]
Letting $N\to\infty$ proves this inequality for the $k_{p,\eta r_0+(1-\eta)r_1}$
and every rational $\eta\in[0,1]$. Then its continuity with $r$ extends the result to every
$\eta\in[0,1]$. Therefore,
\begin{equation*}
k_{p,\eta r_0+(1-\eta)r_1}(x)
\geq
\eta k_{p,r_0}(x)
+
(1-\eta)k_{p,r_1}(x).
\end{equation*}
To summarize, $r\mapsto k_{p,r}(x)$ is finite, nondecreasing, $1$-Lipschitz, and concave.

% By construction, the common core
% $\mathcal M^{\overline\otimes n}$ is dense in both $D_{p,n}$ and
% $C_{p,n}$.  Moreover, for $1<p<\infty$, $D_{p,n}$ is reflexive, and
% the two isometric dualities
% \[
%     C_{p,n}^*=C_{p',n},
%     \qquad
%     C_{p',n}^*=C_{p,n}
% \]
% imply that $C_{p,n}$ is reflexive as well.  Thus the hypotheses of
% Calder\'on duality are satisfied, and
% \cite[Corollary~4.5.2]{BerghLofstrom1976interpolation} gives isometrically
% \[
%     [D_{p,n},C_{p,n}]_\theta^*
%     =
%     [D_{p',n},C_{p',n}]_\theta.
% \]

\begin{lemma}
\label{lem:interpolation-tensorization}
For every $n,m\geq1$, $0\leq\theta\leq1$, let $U \in K_{p,n,t}$
and $V\in K_{p,m,t}$. Then
\begin{equation}
\|U\otimes V\|_{[D_{p,n+m},C_{p,n+m}]_\theta}
=
\|U\|_{[D_{p,n},C_{p,n}]_\theta}
\|V\|_{[D_{p,m},C_{p,m}]_\theta}.
\label{eq:interpolation-tensorization}
\end{equation}
\end{lemma}

\begin{proof}
Let $T(U,V)=U\otimes V$ be the multilinear map. With the tensoring multiplicativity of $D_{p,n}$ and $C_{p,n}$,
\cite[Theorem~4.4.1]{BerghLofstrom1976interpolation} therefore gives
\begin{align}
\|U\otimes V\|_{[D_{p,n+m},C_{p,n+m}]_\theta}
\leq
\|U\|_{[D_{p,n},C_{p,n}]_\theta}
\|V\|_{[D_{p,m},C_{p,m}]_\theta}.
\label{eq:interpolation-tensorization-upper}
\end{align}

For the reverse inequality, since $D_{p,n}$ and $C_{p,n}$ are reflexive, Calderón duality
\cite[Corollary~4.5.2]{BerghLofstrom1976interpolation} gives isometrically
\begin{equation*}
[D_{p,n},C_{p,n}]_\theta^*
=
[D_{p',n},C_{p',n}]_\theta .
\end{equation*}
For nonzero $U$ and $V$, choose $G$ and $H$ such that
\begin{align*}
\|G\|_{[D_{p',n},C_{p',n}]_\theta}
&=
\|H\|_{[D_{p',m},C_{p',m}]_\theta}
=1,\\
|\langle U,G\rangle|
&=
\|U\|_{[D_{p,n},C_{p,n}]_\theta},\\
|\langle V,H\rangle|
&=
\|V\|_{[D_{p,m},C_{p,m}]_\theta}.
\end{align*}
Applying the preceding bilinear interpolation argument to the dual
endpoint spaces yields
\[
\|G\otimes H\|_{[D_{p',n+m},C_{p',n+m}]_\theta}\leq1.
\]
Consequently,
\begin{align*}
\|U\otimes V\|_{[D_{p,n+m},C_{p,n+m}]_\theta}\geq
|\langle U\otimes V,G\otimes H\rangle|=
\|U\|_{[D_{p,n},C_{p,n}]_\theta}
\|V\|_{[D_{p,m},C_{p,m}]_\theta}.
\end{align*}
Together with \eqref{eq:interpolation-tensorization-upper}, this proves
\eqref{eq:interpolation-tensorization}.
\end{proof}

The following theorem gives the asymptotic rate of the $K$-functional.
\begin{theorem}
\label{thm:K-functional-rate}
Let $1<p\leq2$ and let $x\in K_{p,n,t}$ be nonzero. For every $r\in\mathbb R$,
$$ k_{p,r}(x)=\inf_{0\leq \theta \leq 1}\{(1-\theta)r+\log \|x\|_{[D_{p,1},C_{p,1}]_\theta}\}. $$
\end{theorem}
\begin{proof}
Fix any $x\in\mathcal M$, applying Lemma~\ref{lem:interpolation-tensorization} to $x^{\otimes n}$ repeatedly gives
$$ \Vert x^{\otimes n} \Vert_{[D_{p,n},C_{p,n}]_\theta}=\Vert x \Vert_{[D_{p,1},C_{p,1}]_\theta}^n.$$
For fixed $t>0$, consider the Banach space with norm $K_{p,n,t}(\cdot)$.  The identity maps from the norm $D_{p,n}(\cdot)$ and $C_{p,n}(\cdot)$
into $K_{p,n,t}(\cdot)$ satisfy
\begin{align*}
  \Vert I \Vert_{D_{p,n}\to K_{p,n,t}}\leq t,\quad\Vert I \Vert_{C_{p,n}\to K_{p,n,t}}&\leq 1.
\end{align*}
Complex interpolation (cf. Theorem~\ref{complex-interpolation}) therefore gives
$$ \Vert I \Vert_{[D_{p,n},C_{p,n}]_\theta\to K_{p,n,t}}\leq t^{1-\theta}. $$
Hence we obtain
\[
K_{p,n,t}(x^{\otimes n})
\leq
t^{1-\theta}\|x^{\otimes n}\|_{[D_{p,n},C_{p,n}]_{\theta}}=t^{1-\theta}\Vert x \Vert_{[D_{p,1},C_{p,1}]_\theta}^n.
\]
Choose $t=2^{nr}$. It follows that
$$ k_{p,n,r}(x)
=
\frac{1}{n}
\log
K_{p,n,2^{nr}}(x^{\otimes n})\leq (1-\theta)r+\log \Vert x \Vert_{[D_{p,1},C_{p,1}]_\theta}.$$
Passing to the limit $n\to\infty$ and taking the supremum over $r$ gives
\begin{equation}\label{upper-bound}
S_{1-\theta}
:=
\sup_{r\in\mathbb R}
\{k_{p,r}(x)-(1-\theta)r\}
\leq
\log \Vert x \Vert_{[D_{p,1},C_{p,1}]_\theta}.
\end{equation}

For the reverse inequality, the standard embedding
of the real interpolation space into the complex interpolation space (cf.
\cite[Theorem~4.7.1]{BerghLofstrom1976interpolation}) gives
\begin{align}\label{up-inter}
\|x\|_{[D_{p,1},C_{p,1}]_\theta}^n\leq
B_{\theta}
\int_0^\infty
t^{-(1-\theta)}K_{p,n,t}(x^{\otimes n})\frac{dt}{t}
=
n\ln 2B_{\theta}
\int_{\mathbb R}
2^{
n[k_{p,n,r}(x)-(1-\theta) r]
}dr
\end{align}
where $B_{\theta}$ depends only on $\theta$.

Taking two trivial decompositions $x^{\otimes n}=x^{\otimes n}+0$ and $x^{\otimes n}=0+x^{\otimes n}$ in the definition of $K_{p,n,t}(x^{\otimes n})$ gives
\begin{equation*}
k_{p,n,r}(x)-(1-\theta) r
\leq
\min\left\{
\log D_{p,1}(x)+\theta r,
\log C_{p,1}(x)-(1-\theta) r
\right\}.
\end{equation*}
Since $k_{p,n,r}\leq k_{p,r}$, the left-hand side is at most $S_{1-\theta}$. For the right-hand side, it attains its maximum at
$$r^\sharp=\log C_{p,1}(x)-\log D_{p,1}(x),$$ which gives
$$ S_{1-\theta}\leq(1-\theta)\log D_{p,1}(x)+\theta\log C_{p,1}(x).$$
For $0<\theta<1$, set
\begin{equation*}
r_-
:=
\frac{S_{1-\theta}-\log D_{p,1}(x)}{1-(1-\theta)},
\qquad
r_+
:=
\frac{\log C_{p,1}(x)-S_{1-\theta}}{1-\theta},
\end{equation*}
where $r_-\leq r_+$. It follows that
\begin{equation*}
k_{p,n,r}(x)-(1-\theta)r\leq\left\{
\begin{aligned}
&S_{1-\theta}+\theta(r-r_-),\quad r\leq r_-,\\
&S_{1-\theta},\quad r_-\leq r\leq r_+,\\
&S_{1-\theta}-(1-\theta)(r-r_+),\quad r\geq r_+
\end{aligned}\right.
\end{equation*}

\begin{figure}[htbp]
\centering
\begin{tikzpicture}[
    x=1.25cm,
    y=1.35cm,
    line cap=round,
    line join=round,
    font=\small
]

% Numerical positions used only for the schematic.
\def\rminus{-1.3}
\def\rplus{1.4}
\def\Svalue{1.6}
\def\thetavalue{0.4}

% Intersection of the blue and green lines.
\pgfmathsetmacro{\rstar}{
    \thetavalue*\rminus+(1-\thetavalue)*\rplus
}
\pgfmathsetmacro{\Sstar}{
    \Svalue
    +\thetavalue*(1-\thetavalue)*(\rplus-\rminus)
}

% Plot region.
\begin{scope}
\clip (-4,0) rectangle (4,3);

% Three regions.
\fill[blue!6]
(-4,0) rectangle (\rminus,3);

\fill[orange!6]
(\rminus,0) rectangle (\rplus,3);

\fill[green!6]
(\rplus,0) rectangle (4,3);

% Continuations of the three bounds.
\draw[blue!35,thin,domain=-4:4,smooth,variable=\r]
plot
({\r},{\Svalue+\thetavalue*(\r-\rminus)});

\draw[orange!40,thin]
(-4,\Svalue)--(4,\Svalue);

\draw[green!35,thin,domain=-4:4,smooth,variable=\r]
plot
({\r},{\Svalue-(1-\thetavalue)*(\r-\rplus)});

% The effective piecewise upper bound.
\draw[blue,very thick,domain=-4:\rminus,smooth,variable=\r]
plot
({\r},{\Svalue+\thetavalue*(\r-\rminus)});

\draw[orange,very thick]
(\rminus,\Svalue)--(\rplus,\Svalue);

\draw[green!70!black,very thick,domain=\rplus:4,smooth,variable=\r]
plot
({\r},{\Svalue-(1-\thetavalue)*(\r-\rplus)});

% Vertical dividing lines.
\draw[dashed,gray]
(\rminus,0)--(\rminus,3);

\draw[dashed,gray]
(\rplus,0)--(\rplus,3);

% Dashed projections of the blue--green intersection.
\draw[dashed,gray]
(\rstar,0)--(\rstar,\Sstar);

\draw[dashed,gray]
(-4,\Sstar)--(\rstar,\Sstar);

\end{scope}

% Axes.
\draw[->]
(-4.2,0)--(4.3,0)
node[right] {$r$};

\draw[->]
(-4,0)--(-4,3.15);

% Original intersection points.
\fill[blue]
(\rminus,\Svalue) circle (1.8pt);

\fill[green!70!black]
(\rplus,\Svalue) circle (1.8pt);

% Intersection of the blue and green lines.
\fill[black]
(\rstar,\Sstar) circle (1.8pt);

% Axis labels.
\node[below] at (\rminus,0)
{$r_-$};

\node[below] at (\rplus,0)
{$r_+$};

\node[left] at (-4,\Svalue)
{$S_{1-\theta}$};

% Coordinates of the blue--green intersection.
\node[below,align=center] at (\rstar,0)
{$r^\sharp$};

\node[left,align=right,font=\scriptsize] at (-4,\Sstar)
{Max};

% Region labels.
\node[blue,align=center] at (-2.55,0.95)
{$S_{1-\theta}+\theta(r-r_-)$};

\node[orange!80!black,align=center] at (0.05,1.82)
{$S_{1-\theta}$};

\node[green!60!black,align=center] at (2.75,0.75)
{$S_{1-\theta}-(1-\theta)(r-r_+)$};

\end{tikzpicture}
\caption{Upper bound for
$k_{p,n,r}(x)-(1-\theta)r$.}
\label{fig:linear-tail-bound}
\end{figure}
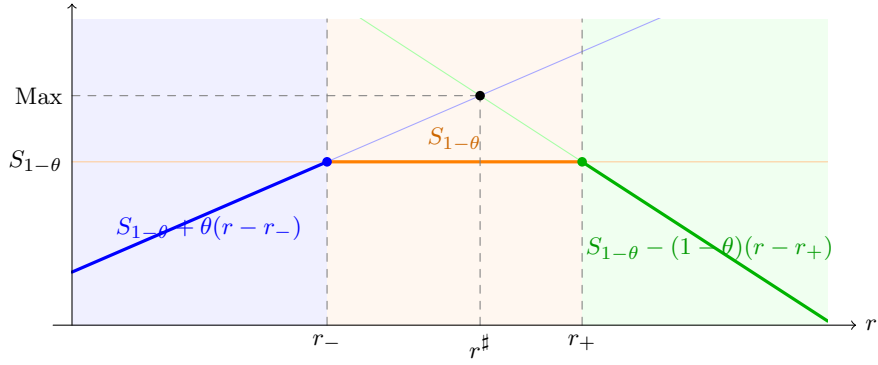

Splitting the
integral at $r_-$ and $r_+$ gives
\begin{align*}
\int_{-\infty}^{r_-}+\int_{r_-}^{r_+}+\int_{r_+}^{\infty}
2^{
n[k_{p,n,r}(x)-(1-\theta) r]
}dr
\leq
2^{nS_{1-\theta}}
\left(
\frac1{n\theta\ln 2}+r_+-r_-
+
\frac1{n(1-\theta)\ln 2 }
\right).
\end{align*}
Substituting this into \eqref{up-inter}, taking
$n^{-1}\log$, and letting $n\to\infty$ gives
\begin{equation*}
\log \|x\|_{[D_{p,1},C_{p,1}]_\theta}
\leq
S_{1-\theta}.
\end{equation*}
Together with \eqref{upper-bound}, this proves
\begin{equation*}
\log \|x\|_{[D_{p,1},C_{p,1}]_\theta}
=S_{1-\theta}=
\sup_{r\in\mathbb R}
\{k_{p,r}(x)-(1-\theta)r\}.
\end{equation*}
for $0<\theta<1$.

When
$\theta=0$, trivial decomposition $x^{\otimes n}=x^{\otimes n}+0$ gives
$$\sup_{r\in\mathbb R}
\{k_{p,r}(x)-(1-0)r\}\leq \log D_{p,1}(x).$$
When $r<0$, the rate satisfies $0<t=2^{nr}<1$. Then \eqref{small-t} gives
$$k_{p,r}(x)=r+\log D_{p,1}(x).$$
Hence we obtain
\begin{equation*}
\sup_{r\in\mathbb R}
\{k_{p,r}(x)-r\}=\log D_{p,1}(x)=\log \|x\|_{[D_{p,1},C_{p,1}]_0}.
\end{equation*}
If $\theta=1$,  choose $\varepsilon_{t}\to 0^+$ and the decomposition $x=y_t+z_t$ such that
\begin{align*}
K_{p,1,t}(x)\leq tD_{p,1}(y_t)+C_{p,1}(z_t)\leq K_{p,1,t}(x)+\varepsilon_t.
\end{align*}
With trivial decomposition $x=0+x$, we have
\begin{align}\label{estimate}
tD_{p,1}(y_t)+C_{p,1}(z_t)\leq C_{p,1}(x)+\varepsilon_t.
\end{align}
Hence
$$0\leq\lim_{t\to \infty}D_{p,1}(y_t)\leq \lim_{t\to \infty}\frac{C_{p,1}(x)+\varepsilon_t}{t}=0$$
and $z_t=x-y_t$ converges to $x$ with respect to the norm $D_{p,1}(\cdot)$.
While \eqref{estimate} also implies
$$ C_{p,1}(z_t)\leq C_{p,1}(x)+\varepsilon_t,$$
one obtains that $z_t$ is uniformly bounded in $C_{p,1}$. Since $C_{p,1}$ is reflexive, there exists a subsequence $(z_{t_n})$ that weakly converges to some $z$ in $C_{p,1}$ and hence in $D_{p,1}$. Since the norm convergence $z_t\to x$ in $D_{p,1}$ implies $z_{t_n}\overset{w}{\to} x$ in $D_{p,1}$, we obtain $z=x$. The weak convergence $z_{t_n}\overset{w}{\to} x$ in $C_{p,1}$ implies that
$$ C_{p,1}(x)\leq \liminf_{t\to \infty} C_{p,1}(z_{t_n})\leq \liminf_{t\to \infty} K_{p,1,t}(x)+\varepsilon_t.$$
While $\varepsilon_t \to 0^+$, this yields
$$ \sup_{r\in\mathbb R}
k_{p,r}(x)\geq \lim_{r\to \infty}
k_{p,1,r}(x)\geq\log C_{p,1}(x).
$$
Besides, the trivial decomposition $x^{\otimes n}=0+x^{\otimes n}$ gives
$$K_{p,n,t}(x^{\otimes n})\leq C_{p,n}(x^{\otimes n})=C_{p,1}(x)^n.$$
 Thus we have
$$ \sup_{r\in\mathbb R}
k_{p,r}(x)\leq \log C_{p,1}(x)$$
and it follows that
$$\sup_{r\in\mathbb R}
k_{p,r}(x)=\log C_{p,1}(x)=\log \|x\|_{[D_{p,1},C_{p,1}]_1}.$$
Now we can conclude that
\begin{equation*}
\log \|x\|_{[D_{p,1},C_{p,1}]_\theta}
=
\sup_{r\in\mathbb R}
\{k_{p,r}(x)-(1-\theta)r\}.
\end{equation*} holds on $0\leq\theta\leq1$.
Since $k_{p,r}$ is continuous, concave and every supergradient lies in
$[0,1]$, the concave Fenchel--Moreau theorem (\cite[Theorem~12.2 and
Corollary~13.3.3]{ConvexAnalysis}) yields
\begin{equation*}
k_{p,r}(x)
=
\inf_{0\leq\theta\leq1}
\left\{
(1-\theta)r+\log \|x\|_{[D_{p,1},C_{p,1}]_{\theta}}
\right\}.
\end{equation*}
This completes the proof.
\end{proof}

\subsection{Achievability Part for $\alpha\in[\frac{1}{2},1)$ }
In this subsection, we establish the achievability part for the strong converse exponent of quantum soft covering, measured by the sandwiched R{\'e}nyi divergence with order $\alpha\in[\frac{1}{2},1)$.

\begin{proposition}\label{prop:renyi}
\label{prop:sce-converse}Let $\alpha\in[\frac{1}{2},1)$ and $a\in(\alpha,1)$. For the codebook-induced
state $\rho_{E}^{\mathcal{{C}}}$ and the target marginal state \(\rho_E\),
we have
\begin{equation*}
\frac{a}{\alpha(1-a)}\log\mathbb{E}_{\mathcal{C}}Q_{\alpha}\left(\rho_{E}^{\mathcal{C}}\|\rho_{E}\right)\leq\log M - \widetilde{I}_a^{\frac{\alpha-a}{\alpha(1-a)}}(X:E)_{\rho_{XE}}.
\end{equation*}
\end{proposition}
\begin{proof}
Let $\sigma_E$ be the optimizing state for $\widetilde{I}_a^{\frac{\alpha-a}{\alpha(1-a)}}(X:E)_{\rho_{XE}}$. We have
\begin{align*}
& \frac{a}{\alpha(1-a)}\log\mathbb{E}_{\mathcal{C}}Q_{\alpha}\left(\rho_{E}^{\mathcal{C}}\|\rho_{E}\right)\nb\\
= & \frac{a}{\alpha(1-a)}\log\mathbb{E}_{\mathcal{C}}\left\Vert \rho_{E}^{\frac{1-\alpha}{2\alpha}}\rho_{E}^{\mathcal{C}}\rho_{E}^{\frac{1-\alpha}{2\alpha}}\right\Vert _{\alpha}^{\alpha}\nb\\
\overset{(a)}{=} & \frac{a}{\alpha(1-a)}\log\mathbb{E}_{\mathcal{C}}\left\Vert \sigma_E^{\frac{a-\alpha}{2a\alpha}} \sigma_E^{\frac{\alpha-a}{2a\alpha}}  \rho_{E}^{\frac{1-\alpha}{2\alpha}}\rho_{E}^{\mathcal{C}}\rho_{E}^{\frac{1-\alpha}{2\alpha}}\sigma_E^{\frac{\alpha-a}{2a\alpha}}\sigma_E^{\frac{a-\alpha}{2a\alpha}}\right\Vert _{\alpha}^{\alpha}\nb\\
\stackrel{(b)}{\leq} & \frac{a}{\alpha(1-a)}\log\mathbb{E}_{\mathcal{C}}\left\Vert \sigma_E^{\frac{a-\alpha}{2a\alpha}}\right\Vert _{\beta}^{2\alpha}\left\Vert \sigma_E^{\frac{\alpha-a}{2a\alpha}}  \rho_{E}^{\frac{1-\alpha}{2\alpha}}\rho_{E}^{\mathcal{C}}\rho_{E}^{\frac{1-\alpha}{2\alpha}}\sigma_E^{\frac{\alpha-a}{2a\alpha}}\right\Vert _{a}^{\alpha}\nb\\
= & \frac{2a}{1-a}\log\left\Vert \sigma_E^{\frac{a-\alpha}{2a\alpha}}\right\Vert _{\beta}+\frac{a}{\alpha(1-a)}\log\mathbb{E}_{\mathcal{C}}\left\Vert \sigma_E^{\frac{\alpha-a}{2a\alpha}}  \rho_{E}^{\frac{1-\alpha}{2\alpha}}\rho_{E}^{\mathcal{C}}\rho_{E}^{\frac{1-\alpha}{2\alpha}}\sigma_E^{\frac{\alpha-a}{2a\alpha}}\right\Vert _{a}^{\alpha}\nb\\
= & \frac{2a}{1-a}\log\left\Vert \sigma_{E}^{\frac{1}{\beta}}\right\Vert _{\beta}+\frac{a}{\alpha(1-a)}\log\mathbb{E}_{\mathcal{C}}\left({\rm Tr}\left(M^{-1}\sigma_E^{\frac{\alpha-a}{2a\alpha}}  \rho_{E}^{\frac{1-\alpha}{2\alpha}}\sum_{m=1}^M\rho_{E}^{X(m)}\rho_{E}^{\frac{1-\alpha}{2\alpha}}\sigma_E^{\frac{\alpha-a}{2a\alpha}}\right)^{a}\right)^{\frac{\alpha}{a}}\nb\\
\stackrel{(c)}{\leq} & \frac{a}{\alpha(1-a)}\log\left(\mathbb{E}_{\mathcal{C}}{\rm Tr}M^{-a}\left(\sigma_E^{\frac{\alpha-a}{2a\alpha}}  \rho_{E}^{\frac{1-\alpha}{2\alpha}}\sum_{m=1}^M\rho_{E}^{X(m)}\rho_{E}^{\frac{1-\alpha}{2\alpha}}\sigma_E^{\frac{\alpha-a}{2a\alpha}}\right)^{a}\right)^{\frac{\alpha}{a}}\nb\\
= & \frac{1}{1-a}\log\mathbb{E}_{\mathcal{C}}{\rm Tr}M^{-a}\left(\sigma_E^{\frac{\alpha-a}{2a\alpha}}  \rho_{E}^{\frac{1-\alpha}{2\alpha}}\sum_{m=1}^M\rho_{E}^{X(m)}\rho_{E}^{\frac{1-\alpha}{2\alpha}}\sigma_E^{\frac{\alpha-a}{2a\alpha}}\right)^{a}\nb\\
\stackrel{(d)}{\leq} & \frac{1}{1-a}\log\mathbb{E}_{\mathcal{C}}{\rm Tr}M^{-a}\sum_{m=1}^M\left(\sigma_E^{\frac{\alpha-a}{2a\alpha}}  \rho_{E}^{\frac{1-\alpha}{2\alpha}}\rho_{E}^{X(m)}\rho_{E}^{\frac{1-\alpha}{2\alpha}}\sigma_E^{\frac{\alpha-a}{2a\alpha}}\right)^{a}\nb\\
= & \log M+\frac{1}{1-a}\log{\rm Tr}\sum_{x\in\mathcal{X}}P_X(x)\left(\sigma_E^{\frac{\alpha-a}{2a\alpha}}  \rho_{E}^{\frac{1-\alpha}{2\alpha}}\rho_{E}^{x}\rho_{E}^{\frac{1-\alpha}{2\alpha}}\sigma_E^{\frac{\alpha-a}{2a\alpha}}\right)^{a} \nb\nb\\
= & \log M+\frac{1}{1-a}\log{\rm Tr}\left((\rho_X\otimes\sigma_E)^{\frac{\alpha-a}{2a\alpha}}  (\rho_X\otimes\rho_{E})^{\frac{1-\alpha}{2\alpha}}\rho_{XE}(\rho_X\otimes\rho_{E})^{\frac{1-\alpha}{2\alpha}}(\rho_X\otimes\sigma_E)^{\frac{\alpha-a}{2a\alpha}}\right)^{a} \nb\nb\\
=& \log M - \widetilde{I}_a^{\frac{\alpha-a}{\alpha(1-a)}}(X:E)_{\rho_{XE}},
\end{align*}
where $(a)$ is because $\supp(\rho_{E})\subset\supp(\sigma_{E})$,
$(b)$ is due to H{\"o}lder's inequality with $\frac{2}{\beta}+\frac{1}{a}=\frac{1}{\alpha}$,
$(c)$ comes from Jensen's inequality
and $x\mapsto x^s$ is concave function with $s\in[0,1]$, and $(d)$
follows from McCarthy's inequality.
\end{proof}
\begin{theorem}\label{Thm:channel-res-renyi}
Let $\rho_{XE}=\sum_{x\mathcal{\in X}}P_{X}(x)|x\rangle\langle x|\otimes\rho_{E}^{x}$
be a C-Q state and rate $R\geq0$. Then, for any $\alpha\in[\frac{1}{2},1)$, we have
\begin{equation*}
\varGamma_{\rm sc}^{(\alpha)}(\rho_{XE},R)\geq\sup_{\alpha\leq a\leq1}\frac{\alpha\left(1-a\right)}{a\left(1-\alpha\right)}\left\{\widetilde{I}_a^{\frac{\alpha-a}{\alpha(1-a)}}(X:E)_{\rho_{XE}}-R\right\}.
\end{equation*}
\end{theorem}
\begin{proof}
Let $\sigma_E$ be the optimizing state for $\widetilde{I}_a^{\frac{\alpha-a}{\alpha(1-a)}}(X:E)_{\rho_{XE}}$ and $\mathcal{C}_{n}:=\{ X^{n}(1),...,X^{n}(2^{nR})\} $ be
a random i.i.d. codebook,
where each $X^{n}(m)$ is independently drawn from $P_{X}^{\otimes n}$.
Proposition \ref{prop:sce-converse} gives
\begin{align*}
&\frac{a}{\alpha(1-a)}\log\mathbb{E}_{\mathcal{C}_n}Q_{\alpha}\left(\rho_{E^n}^{\mathcal{C}_n}\|\rho_{E}^{\otimes n}\right) \nb\\
\leq & nR+\frac{1}{1-a}\log{\rm Tr}\sum_{x^n\in\mathcal{X}^n}P_X^{\otimes n}(x^n)\left((\sigma_E^{\otimes n})^{\frac{\alpha-a}{2a\alpha}}  (\rho_{E}^{\otimes n})^{\frac{1-\alpha}{2\alpha}}\rho_{E^n}^{x^n}(\rho_{E}^{\otimes n})^{\frac{1-\alpha}{2\alpha}}(\sigma_E^{\otimes n})^{\frac{\alpha-a}{2a\alpha}}\right)^{a} \nb\\
= & nR+\frac{1}{1-a}\log{\rm Tr}\left((\rho_X^{\otimes n}\otimes\sigma_E^{\otimes n})^{\frac{\alpha-a}{2a\alpha}}  (\rho_X^{\otimes n}\otimes\rho_{E}^{\otimes n})^{\frac{1-\alpha}{2\alpha}}\rho_{XE}^{\otimes n}(\rho_X^{\otimes n}\otimes\rho_{E}^{\otimes n})^{\frac{1-\alpha}{2\alpha}}(\rho_X^{\otimes n}\otimes\sigma_E^{\otimes n})^{\frac{\alpha-a}{2a\alpha}}\right)^{a} \nb\\
=&nR -n \widetilde{I}_a^{\frac{\alpha-a}{\alpha(1-a)}}(X:E)_{\rho_{XE}}.
\end{align*}
Since the above inequality holds for any $a\in[\alpha,1)$, we have
\begin{align*}
& \liminf_{n\rightarrow\infty}\frac{1}{n}D_{\alpha}(\rho_{\mathcal{C}_{n}E^{n}}\|\rho_{\mathcal{C}_{n}}\otimes\rho_{E}^{\otimes n})\nb\\
= & \liminf_{n\rightarrow\infty}\frac{1}{n(\alpha-1)}\log\mathbb{E}_{\mathcal{C}_n}Q_{\alpha}\left(\rho_{E^n}^{\mathcal{C}_n}\|\rho_{E}^{\otimes n}\right)\nb\\
\geq & \sup_{\alpha\leq a\leq1}\frac{\alpha\left(1-a\right)}{a\left(1-\alpha\right)}\left\{\widetilde{I}_a^{\frac{\alpha-a}{\alpha(1-a)}}(X:E)_{\rho_{XE}}-R\right\}.
\end{align*}
\end{proof}

\subsection{Optimality Part for $\alpha\in[\frac{1}{2},1)$}
We next prove the optimality part for the strong converse exponent of the quantum soft covering. Together
with Theorem~\ref{Thm:channel-res-renyi}, this gives the exact strong converse exponent for $\alpha\in[\frac{1}{2},1)$. Throughout this subsection, we make the convention that $\rho_{XE}=\sum_{x\in\mathcal{X}} P_X(x)\, |x\rangle\langle x| \otimes \rho_E^x$
be a C-Q state and $
\mathcal C=\{X_1,\ldots,X_M\}$
be a random codebook whose codewords are independently drawn according
to $P_X$. For any $\mathcal X$-indexed family of
matrices $\{A_x\}_{x\in\mathcal X}$ and any random
variable $X$ with distribution $P_X$, we use $A_X$ to denote the
matrix-valued random variable that takes the value $A_x$ with
probability $P_X(x)$.

We first establish the required one-shot lower bound on
$\mathbb E_{\mathcal C}Q_\alpha
(\rho_E^{\mathcal C}\|\rho_E)$.

\begin{proposition}
\label{prop:one-shot-lower-alpha-half-one}
Let $\alpha\in(\frac{1}{2},1)$ and \(X\) be a random variable distributed according to \(P_X\). Define
\begin{equation}
A_{X}
:=
(\rho_E^{X})^{1/2}
\rho_E^{\frac{1-\alpha}{2\alpha}}.
\label{eq:def-Ax-lower}
\end{equation}
Then there exists a constant $c_\alpha>0$, depending only on $\alpha$,
such that
\begin{align*}
\left[
\mathbb E_{\mathcal C}
Q_\alpha\left(
\rho_E^{\mathcal C}\middle\|\rho_E
\right)
\right]^{\frac{1}{2\alpha}}&\geq
c_\alpha
\inf_{A_{X}=U_{X}+V_{X}}
\Bigg\{
M^{\frac{1-\alpha}{2\alpha}}
\left(
\mathbb{E}_{X}\operatorname{Tr}|U_X|^{2\alpha}
\right)^{\frac{1}{2\alpha}}+
\left[
\operatorname{Tr}
\left(
\mathbb{E}V_X^*V_X
\right)^\alpha
\right]^{\frac{1}{2\alpha}}
\Bigg\}.
\label{eq:one-shot-lower-alpha-half-one}
\end{align*}
\end{proposition}

\begin{proof}
Write $e_{m n}=|m\rangle\langle n|$. For the codebook
$\mathcal C=\{X_1,\ldots,X_M\}$, set
\begin{equation*}
Z_{\mathcal C}
:=
\frac1{\sqrt M}\sum_{m=1}^M A_{X_m}\otimes e_{m1}.
\end{equation*}
Since $
A_X^*A_X
=
\rho_E^{\frac{1-\alpha}{2\alpha}}
\rho_E^X
\rho_E^{\frac{1-\alpha}{2\alpha}},
$
we have
\begin{align*}
Z_{\mathcal C}^*Z_{\mathcal C}=
\frac1M
\sum_{m=1}^M A_{X_m}^*A_{X_m}\otimes e_{11}=
\rho_E^{\frac{1-\alpha}{2\alpha}}
\rho_E^{\mathcal C}
\rho_E^{\frac{1-\alpha}{2\alpha}}\otimes e_{11},
\end{align*}
Set $p=2\alpha$ and therefore
\begin{equation*}
\|Z_{\mathcal C}\|_p^p
=
Q_\alpha\left(
\rho_E^{\mathcal C}\middle\|\rho_E
\right).
\end{equation*}

Let $\varepsilon_1,\ldots,\varepsilon_M$ be independent Rademacher
random variables. For any $m\in\{1,2,\cdots,M\}$, define
\begin{equation*}
B_m
:=
\frac1{\sqrt M}\varepsilon_m A_{X_m}.
\end{equation*}
Since $\sum_m\varepsilon_m e_{mm}$ is unitary, one obtains
\begin{align}
\left[
\mathbb E_{\mathcal C}
Q_\alpha\left(
\rho_E^{\mathcal C}\middle\|\rho_E
\right)
\right]^{1/p}=
\left(
\mathbb E_{\mathcal C,\varepsilon}
\left\|
\sum_{m=1}^M  B_m\otimes e_{m1}
\right\|_p^p
\right)^{1/p}.
\label{eq:Q-rademacher-column}
\end{align}
Applying  the matrix Rosenthal inequality (cf. Appendix Theorem~\ref{lem:reverse-matrix-rosenthal}), we obtain
\begin{align}
\left(
\mathbb E_{\mathcal C,\varepsilon}
\left\|
\sum_{m=1}^M B_m\otimes e_{m1} 
\right\|_p^p
\right)^{1/p}&\geq
c_\alpha
\inf_{B_m=D_m+C_m+R_m}
\Bigg\{
\left(
\sum_{m=1}^M
\mathbb E\operatorname{Tr}|D_m|^p
\right)^{1/p}
\nb\\
&\hspace{25mm}
+
\left[
\operatorname{Tr}
\left(
\sum_{m=1}^M
\mathbb E C_m^*C_m
\right)^{p/2}
\right]^{1/p}
\nb\\
&\hspace{25mm}
+
\left[
\sum_{m=1}^M
\operatorname{Tr}
\left(
\mathbb E R_mR_m^*
\right)^{p/2}
\right]^{1/p}
\Bigg\}.
\label{eq:one-column-rosenthal}
\end{align}
Here $D_m,C_m,R_m$ are mean zero elements measurable with respect to
$(X_m,\varepsilon_m)$. In the three terms above, $\mathbb E$ denotes
expectation over $(X_m,\varepsilon_m)$; since each component is local,
this is the same as the full expectation.

We first eliminate the row term. Since $p/2=\alpha<1$, concavity of
$T\mapsto\operatorname{Tr}T^{p/2}$ gives, for every $m$,
\begin{equation*}
\operatorname{Tr}
\left(
\mathbb E R_mR_m^*
\right)^{p/2}
\geq
\mathbb E\operatorname{Tr}
\left(
R_mR_m^*
\right)^{p/2}
=
\mathbb E\operatorname{Tr}|R_m|^p.
\label{eq:row-to-diagonal}
\end{equation*}
Consequently, Minkowski's inequality yields
\begin{align*}
\left(
\sum_m
\mathbb E\operatorname{Tr}|D_m+R_m|^p
\right)^{1/p}&\leq
\left(
\sum_m
\mathbb E\operatorname{Tr}|D_m|^p
\right)^{1/p}
+
\left(
\sum_m
\mathbb E\operatorname{Tr}|R_m|^p
\right)^{1/p}
\nb\\
&\leq
\left(
\sum_m
\mathbb E\operatorname{Tr}|D_m|^p
\right)^{1/p}
+
\left[
\sum_m
\operatorname{Tr}
\left(
\mathbb E R_mR_m^*
\right)^{p/2}
\right]^{1/p}.
\end{align*}
Thus every decomposition $B_m=(D_m+R_m)+C_m$ can be viewed as $B_m=\tilde{D}_m+C_m$.
Therefore, the infimum in \eqref{eq:one-column-rosenthal} is equal to
\begin{align}
\inf_{B_m=D_m+C_m}
\Bigg\{
\left(
\sum_m
\mathbb E\operatorname{Tr}|D_m|^p
\right)^{1/p}+
\left[
\operatorname{Tr}
\left(
\sum_m
\mathbb E C_m^*C_m
\right)^{p/2}
\right]^{1/p}
\Bigg\}.
\label{eq:two-term-column-decomposition}
\end{align}
Here $D_m$ and $C_m$ are still mean zero.

For any $x\in\mathcal{X}$ and decomposition $B_m=D_m+C_m$
 in
\eqref{eq:two-term-column-decomposition}, set
\begin{equation*}
U_{m,x}
=
\frac{\sqrt M}{2}
\left(
D_m(x,+1)-D_m(x,-1)
\right),
\quad
V_{m,x}
=
\frac{\sqrt M}{2}
\left(
C_m(x,+1)-C_m(x,-1)
\right).
\end{equation*}
Since
\begin{equation*}
\frac{\varepsilon}{\sqrt M}A_{x}
=
D_m(x,\varepsilon)+C_m(x,\varepsilon),
\qquad
\varepsilon\in\{-1,+1\},
\end{equation*}
we obtain $A_{x}=U_{m,x}+V_{m,x}$. The convexity of
$T\mapsto\operatorname{Tr}|T|^p$ implies
\begin{align*}
M^{-p/2}\operatorname{Tr}|U_{m,x}|^p
&=
\operatorname{Tr}
\left|
\frac{D_m(x,+1)-D_m(x,-1)}2
\right|^p
\\
&\leq
\frac12\Tr|D_m(x,+1)|^p
+
\frac12\Tr|D_m(x,-1)|^p.
\end{align*}
It then follows
\begin{equation}
M^{-p/2}
\sum_{m,x} P_{X_m}(x)\Tr|U_{m,x}|^p
\leq
\sum_m
\mathbb E_{X_m,\varepsilon_m}\operatorname{Tr}|D_m|^p.
\label{eq:odd-diagonal-contraction}
\end{equation}
Besides, the column part satisfies
\begin{align*}
0\leq& \frac14
\left(
C_m(x,+1)+C_m(x,-1)
\right)^*
\left(
C_m(x,+1)+C_m(x,-1)
\right) \nb\\
=&
\frac12C_m(x,+1)^*C_m(x,+1)
+
\frac12C_m(x,-1)^*C_m(x,-1)
-
M^{-1}V_{m,x}^*V_{m,x}.
\end{align*}
Taking the expectation with respect to $X_m$ and summation over $m$ gives
\begin{equation*}
M^{-1}
\sum_{m,x} P_{X_m}(x)V_{m,x}^*V_{m,x}
\leq
\sum_m
\mathbb E_{X_m,\varepsilon_m} C_m^*C_m.
\label{eq:odd-column-contraction}
\end{equation*}
Now define
\begin{equation}
	U_x
	=
	\frac1M\sum_{m=1}^M U_{m,x},
	\qquad
	V_x
	=
	\frac1M\sum_{m=1}^M V_{m,x}.
	\label{eq:average-local-decompositions}
\end{equation}
Since $A_x=U_{m,x}+V_{m,x}$ for every $m\in\{1,2,\cdots,M\}$ and
$x\in\mathcal X$, averaging over $m$ gives
\[
A_x
=
\frac1M\sum_{m=1}^M\bigl(U_{m,x}+V_{m,x}\bigr)
=
U_x+V_x.
\]
Consequently, for any $m\in\{1,2,\cdots,M\}$, we have $A_{X_m}=U_{X_m}+V_{X_m}$ and $A_{X}=U_{X}+V_{X}$.
By the convexity of the function $X\mapsto\Tr |X|^p$ and
\eqref{eq:odd-diagonal-contraction},
\begin{align*}
M^{\frac1p-\frac12}
\left(
\sum_xP_X(x)\operatorname{Tr}|U_x|^p
\right)^{1/p}&\leq
M^{\frac1p-\frac12}
\left(
\frac1M
\sum_{m,x}
P_{X_m}(x)\operatorname{Tr}|U_{m,x}|^p
\right)^{1/p}
\nb\\
&=
M^{-1/2}
\left(
\sum_{m,x}
P_{X_m}(x)\operatorname{Tr}|U_{m,x}|^p
\right)^{1/p}
\nb\\
&\leq
\left(
\sum_m
\mathbb E\operatorname{Tr}|D_m|^p
\right)^{1/p}.
\end{align*}
Moreover,
\begin{align*}
\sum_xP_X(x)V_x^*V_x\leq
\frac1M
\sum_{m,x}
P_{X_m}(x)V_{m,x}^*V_{m,x}\leq
\sum_m
\mathbb E C_m^*C_m,
\end{align*}
where the first inequality follows from, for each $x$,
\begin{equation*}
\frac1M\sum_mV_{m,x}^*V_{m,x}-V_x^*V_x
=
\frac1M\sum_m
\left(
V_{m,x}-V_x
\right)^*
\left(
V_{m,x}-V_x
\right)
\geq0.
\end{equation*}
Monotonicity of $T\mapsto\operatorname{Tr}T^{p/2}$ on positive
matrices therefore gives
\begin{align*}
\left[
\operatorname{Tr}
\left(
\sum_xP_X(x)V_x^*V_x
\right)^{p/2}
\right]^{1/p}\leq
\left[
\operatorname{Tr}
\left(
\sum_m
\mathbb E C_m^*C_m
\right)^{p/2}
\right]^{1/p}.
\label{eq:column-after-averaging}
\end{align*}
Thus every decomposition in
\eqref{eq:two-term-column-decomposition} has cost at least
\begin{align*}
\inf_{A_X=U_X+V_X}
\Bigg\{
M^{\frac1p-\frac12}
\left(
\mathbb{E}_X\operatorname{Tr}|U_X|^p
\right)^{1/p}+
\left[
\operatorname{Tr}
\left(
\mathbb{E}_X V_X^*V_X
\right)^{p/2}
\right]^{1/p}
\Bigg\}.
\end{align*}
Taking the infimum in
\eqref{eq:two-term-column-decomposition}, using
\eqref{eq:Q-rademacher-column}, and substituting
$1/p-1/2=(1-\alpha)/(2\alpha)$ proves the desired inequality.
\end{proof}

The next lemma gives the exact exponential behavior of the
$K$-functional appearing in
Proposition~\ref{prop:one-shot-lower-alpha-half-one}. 
\begin{lemma}
\label{lem:tensorized-column-interpolation}
Let $\alpha\in(\frac12,1)$ and $p:=2\alpha$. Let \(X\) be a random variable distributed according to \(P_X\). Define
\begin{equation*}
A_X
:=
(\rho_E^X)^{1/2}
\rho_E^{\frac{1-\alpha}{2\alpha}}
\end{equation*}
and $A_{X^n}=A_{X}^{\otimes n}$. Also set
\begin{align*}
D_{q,n}(A_{X^n}):=
\left(
\mathbb{E}
\operatorname{Tr}|A_{X^n}|^q
\right)^{1/q},\quad
C_{p,n}(A_{X^n}):=
\left[
\operatorname{Tr}
\left(
\mathbb{E}
A_{X^n}^*A_{X^n}
\right)^{p/2}
\right]^{1/p}.
\end{align*}
For $t>0$, set
\begin{equation*}
K_n(t)
:=
\inf_{A_{X^n}=U+V}
\left\{
tD_{p,n}(U)+C_{p,n}(V)
\right\},
\end{equation*}
Then, for every $r\in\mathbb R$,
\begin{align}
\lim_{n\to\infty}
\frac{1}{n}\log K_n(2^{nr})
=
\inf_{\alpha\leq a\leq1}
\left\{
\frac{\alpha(1-a)}{a(1-\alpha)}r
-
\frac{1-a}{2a}
\widetilde I_a^{
\frac{\alpha-a}{\alpha(1-a)}
}(X:E)_{\rho_{XE}}
\right\}.
\label{eq:tensor-K-exponent}
\end{align}
\end{lemma}

\begin{proof}
We divide the proof into three steps.

\noindent\emph{\bf Step 1: Amalgamated interpolation.} We may remove the letters $x$ for which $P_X(x)=0$.  For every
$n\geq1$, consider the finite von Neumann algebras
\begin{align*}
\mathcal M_n=
\ell_\infty(\mathcal X^n)
\,\overline\otimes\,
\mathcal B(E^n),\quad
\mathcal N_n
=
\mathcal B(E^n),
\end{align*}
equipped with the traces
\begin{align*}
\tau_{\mathcal M_n}(A_{X^n})=
\mathbb{E} \operatorname{Tr}A_{X^n},\quad
\tau_{\mathcal N_n}(B)
=
\operatorname{Tr}B.
\end{align*}
The faithful trace-preserving conditional expectation from
$\mathcal M_n$ onto $\mathcal N_n$ is
\begin{equation*}
\mathbb E(A_{X^n})
:=
\sum_{x^n}P_X^{\otimes n}(x^n)A_{X^n}
\end{equation*}
In particular,
\[
\|A_{X^n}\|_{L_q(\mathcal M_n)}=D_{q,n}(A_{X^n}),\quad \|A_{X^n}\|_{L_p^{c}(\mathcal M_n,\mathbb{E})}=C_{p,n}(A_{X^n}).
\]

The details of the following properties of amalgamated spaces are included in Appendix~\ref{amag}. Let $s_1$ be determined by
\begin{equation}
\frac1{s_1}
:=
\frac1p-\frac12.
\label{eq:def-s1}
\end{equation}
% We first verify the factorization formula
% \begin{equation}
% C_{p,n}(W)
% =
% \inf_{W_{x^n}=Z_{x^n}B}
% D_{2,n}(Z)\|B\|_{s_1}.
% \label{eq:Cpn-factorization}
% \end{equation}
% By absorbing the unitary in the polar decomposition of $B$
% into $Z_{x^n}$, it suffices to take $B\geq0$.  If
% $W_{x^n}=Z_{x^n}B$ and $
% S_Z:=
% \mathbb{E}Z_{x^n}^*Z_{x^n}$, then H{\"o}lder's inequality gives
% \begin{align*}
% C_{p,n}(W)=
% \|S_Z^{1/2}B\|_p\leq
% \|S_Z^{1/2}\|_2\|B\|_{s_1}
% =
% D_{2,n}(Z)\|B\|_{s_1}.
% \end{align*}
% Conversely, put $
% S_W:=\mathbb{E}W_{x^n}^*W_{x^n}$.
% If $S_W\neq0$, define, on $\operatorname{supp}S_W$,
% \begin{equation*}
% B
% :=
% \left(\frac{S_W^{p/2}}{\operatorname{Tr}S_W^{p/2}}\right)^{1/s_1},
% \qquad
% Z_{x^n}:=W_{x^n}B^{-1}.
% \end{equation*}
% Then $\|B\|_{s_1}=1$ and, using
% (\ref{eq:def-s1}),
% \begin{align*}
% D_{2,n}(Z)^2=
% \operatorname{Tr}B^{-1}S_WB^{-1}=
% \left(
% \operatorname{Tr}S_W^{p/2}
% \right)^{2/s_1+1}
% =
% \left(
% \operatorname{Tr}S_W^{p/2}
% \right)^{2/p}.
% \end{align*}
% Thus $D_{2,n}(Z)=C_{p,n}(W)$. The singular case is similar (
% $W_{x^n}$ vanishes on $\ker S_W$ whenever
% $P_X^{\otimes n}(x^n)>0$). This proves
% \eqref{eq:Cpn-factorization}.
For $a\in(\alpha,1)$, put
\begin{equation*}
q:=2a,
\qquad
r_a
:=
\frac1p-\frac{1}{2a}
=
\frac{a-\alpha}{2a\alpha}.
\end{equation*}
Define
\begin{align}
\Phi_{a,n}(A_{X^n})
:=
\inf_{\substack{\sigma_{E^n}>0\\
\operatorname{Tr}\sigma_{E^n}=1}}
\Bigg[
\mathbb{E}
\operatorname{Tr}
\left(
\sigma_{E^n}^{-r_a}
A_{X^n}^*A_{X^n}
\sigma_{E^n}^{-r_a}
\right)^a
\Bigg]^{1/(2a)}.
\label{eq:def-Phi-an}
\end{align}
We now check the exact factorization
\begin{equation*}
\Phi_{a,n}(A_{X^n})
=
\inf_{A_{X^n}=LZ_{X^n}R}
\|L\|_{\infty} D_{2a,n}(Z)\|R\|_{1/r_a}=\|A_{X^n}\|_{\infty,2a,1/r_a}.
\end{equation*}
To see this, for any decomposition $A_{X^n}=LZ_{X^n}R$, take
\begin{equation*}
\sigma_{E^n}
:=
\|R\|_{1/r_a}^{-1/r_a}|R|^{1/r_a}
\end{equation*}
which is a density operator. With the polar decomposition $R=U|R|$, we obtain
$$ A_{X^n}\sigma_{E^n}^{-r_a}=LZ_{X^n}R\sigma_{E^n}^{-r_a}=\|R\|_{1/r_a}LZ_{X^n}U.$$
This gives 
\begin{align*}
\Phi_{a,n}(A_{X^n})=&\inf_{\substack{\sigma_{E^n}>0\\
\Tr \sigma_{E^n}=1}}
D_{2a,n}(A_{X^n}\sigma_{E^n}^{-r_a})\\
\leq &D_{2a,n}(LZ_{X^n}U)\|R\|_{1/r_a}\\
\leq &\|L\|_\infty D_{2a,n}(Z_{X^n})\|R\|_{1/r_a}.
\end{align*}
Conversely, for every admissible
$\sigma_{E^n}$, take
\begin{equation*}
L:=I_{E^n},\qquad R:=\sigma_{E^n}^{r_a},
\qquad
Z_{X^n}:=A_{X^n}\sigma_{E^n}^{-r_a}.
\end{equation*}
Then $\|R\|_{1/r_a}=1$, and the factorization cost equals the
corresponding expression in \eqref{eq:def-Phi-an}.  This gives $ \Phi_{a,n}(A_{X^n})\geq \|A_{X^n}\|_{\infty,2a,1/r_a}$. An approximation argument treats singular cases.

The two endpoint triples for the amalgamated spaces are
\begin{equation*}
(u_0,q_0,v_0)=(\infty,p,\infty),
\qquad
(u_1,q_1,v_1)=(\infty,2,s_1).
\label{eq:amalgamated-endpoints}
\end{equation*}
They satisfy the hypotheses of Theorem~\ref{amag-inter}, where
\begin{equation*}
\frac1{u_j}+\frac1{q_j}+\frac1{v_j}
=
\frac1p
\leq1,\qquad 2\leq u_j,v_j\leq \infty.
\end{equation*}
Moreover, one has (cf. Appendix~\ref{amag})
\begin{align*}
L_\infty(\mathcal N_n)
L_p(\mathcal M_n)
L_\infty(\mathcal N_n)
=
D_{p,n},\quad
L_\infty(\mathcal N_n)
L_2(\mathcal M_n)
L_{s_1}(\mathcal N_n)
=
C_{p,n}.
\end{align*}
 At $
\theta_a
=
\frac{a-\alpha}{a(1-\alpha)}$,
the interpolated indices satisfy
\begin{equation*}
\frac{1-\theta_a}{p}+\frac{\theta_a}{2}
=
\frac1{2a},
\qquad
\frac{1-\theta_a}{\infty}+\frac{\theta_a}{s_1}=\frac{1}{1/r_a}.
\end{equation*}
Consequently, the amalgamated interpolation theorem gives
\begin{equation}
\|A_{X^n}\|_{[D_{p,n},C_{p,n}]_{\theta_a}}
=
\Phi_{a,n}(A_{X^n}).
\label{eq:exact-interpolation-Phi}
\end{equation}
At $a=\alpha$(or $a=1$), this identity means
$\Phi_{\alpha,n}(A_{X^n})=D_{p,n}(A_{X^n})$
(or $\Phi_{1,n}(A_{X^n})=C_{p,n}(A_{X^n})$).

% \noindent\emph{\bf Step 2: multiplicativity on the tensor powers of $A$.} Fix $a\in(\alpha,1)$ and recall that
% \begin{equation*}
% -2r_a
% =
% \frac{\alpha-a}{a\alpha}
% \in(-1,0).
% \end{equation*}

% Applying Lemma~\ref{lem:interpolation-tensorization} with
% $\theta=\theta_a$ and using \eqref{eq:exact-interpolation-Phi}, we obtain
% \begin{equation*}
% \Phi_{a,n+m}(U\otimes V)
% =
% \Phi_{a,n}(U)\Phi_{a,m}(V).
% \end{equation*}
% Hence we get $
% \Phi_{a,n}(A^{\otimes n})
% =
% \Phi_{a,1}(A)^n$.

% At $a=\alpha$ and $a=1$, the same identity follows directly from the
% tensor multiplicativity of $D_{p,n}$ and $C_{p,n}$.

Besides, we have
\begin{equation*}
\frac{\alpha-a}{\alpha(1-a)}\cdot \frac{1-a}{2a}
=
-r_a,
\qquad
\left(1-\frac{\alpha-a}{\alpha(1-a)}\right)\frac{1-a}{2a}
=
\frac{1-\alpha}{2\alpha}.
\end{equation*}
The definition of the club-sandwiched information gives
\begin{align*}
\widetilde I_a^{\frac{\alpha-a}{\alpha(1-a)}}(X:E)_{\rho_{XE}}
&=
\frac1{a-1}
\log
\inf_{\sigma_E\in\mathcal D(E)}
\sum_xP_X(x)
\operatorname{Tr}
\left(
\sigma_E^{-r_a}A_x^*A_x\sigma_E^{-r_a}
\right)^a.
\end{align*}
Consequently,
\begin{equation}
\|A_X\|_{[D_{p,1},C_{p,1}]_{\theta_a}}
=
\Phi_{a,1}(A_X)
=
\exp\left\{
-\frac{1-a}{2a}
\widetilde I_a^{
\frac{\alpha-a}{\alpha(1-a)}
}(X:E)_{\rho_{XE}}
\right\}.
\label{eq:Phi-information-identification}
\end{equation}
The endpoint values follow by continuity; in particular,
\begin{equation*}
\Phi_{1,1}(A_X)
=
C_{p,1}(A_X)
=
\left[
\operatorname{Tr}
\left(
\rho_E^{1/\alpha}
\right)^\alpha
\right]^{1/(2\alpha)}
=1.
\end{equation*}

\noindent\emph{\bf Step 2: the $K$-functional.} We now put
$$k_n(r):=\frac{1}{n}\log K_n(2^{nr}),\quad k(r):=\lim_{n\to\infty} k_n(r).$$
Thus, by the discussion in Section~\ref{K-func}, $k(r)$ is finite, nondecreasing, $1$-Lipschitz, and concave. Theorem~\ref{thm:K-functional-rate} then yields
\begin{equation}
k(r)
=
\inf_{0\leq\theta_a\leq1}
\left\{
(1-\theta_a)r+\log \Phi_{a,1}(A_X)
\right\}.
\label{eq:K-Fenchel-inversion}
\end{equation}
Finally, $a\mapsto\theta_a$ is a bijection from $[\alpha,1]$ onto
$[0,1]$.  Substituting
\eqref{eq:exact-interpolation-Phi}, and
\eqref{eq:Phi-information-identification} into
\eqref{eq:K-Fenchel-inversion} proves
\eqref{eq:tensor-K-exponent}.
\end{proof}

\begin{theorem}[Exact strong converse exponent]
\label{thm:exact-sc-alpha-half-one}
Let
\begin{equation*}
\rho_{XE}
=
\sum_{x\in\mathcal X}
P_X(x)|x\rangle\langle x|\otimes\rho_E^x,
\end{equation*}
and let $R\geq0$.  For every
$\alpha\in[\frac12,1)$,
\begin{align}
\varGamma_{\rm sc}^{(\alpha)}
(\rho_{XE},R)
&=
\sup_{\alpha\leq a\leq1}
\frac{\alpha(1-a)}{a(1-\alpha)}
\left\{
\widetilde I_a^{
\frac{\alpha-a}{\alpha(1-a)}
}(X:E)_{\rho_{XE}}
-R
\right\}.
\label{eq:exact-sc-alpha-half-one}
\end{align}
\end{theorem}

\begin{proof}
We first consider $\alpha\in(\frac12,1)$.  Let
$\mathcal C_n=\{X^n(1),\ldots,X^n(M_n)\}$ be an i.i.d. random
codebook with $M_n=2^{nR}$, where every codeword is distributed according to
$P_X^{\otimes n}$ and $
M=2^R.$
The matrices in \eqref{eq:def-Ax-lower} tensorize
\begin{equation*}
A_{X^n}
=
A_{X_1}\otimes\cdots\otimes A_{X_n}.
\end{equation*}
Set $
t_n
:=
M_n^{\frac{1-\alpha}{2\alpha}}$ and $r_n:=\frac{1}{n}\log t_n.$
Then
\begin{equation*}
r_n
=
r_0
:=
\frac{1-\alpha}{2\alpha}R.
\end{equation*}
Proposition~\ref{prop:one-shot-lower-alpha-half-one}, applied to the
$n$-fold state, gives
\begin{equation}
\left[
\mathbb E_{\mathcal C_n}
Q_\alpha
\left(
\rho_{E^n}^{\mathcal C_n}
\middle\|
\rho_E^{\otimes n}
\right)
\right]^{1/(2\alpha)}
\geq
c_\alpha K_n(t_n).
\label{eq:Q-lower-Kn}
\end{equation}
Besides, Lemma~\ref{lem:tensorized-column-interpolation} implies
\begin{align*}
\lim_{n\to\infty}
\frac{1}{n}\log K_n(t_n)
=
\inf_{\alpha\leq a\leq1}
\frac{1-a}{2a}
\left\{
R
-
\widetilde I_a^{
\frac{\alpha-a}{\alpha(1-a)}
}(X:E)_{\rho_{XE}}
\right\},
\end{align*}
where we used
\begin{equation*}
\frac{\alpha(1-a)}{a(1-\alpha)}r_0
=
\frac{1-a}{2a}R.
\end{equation*}
Hence, taking logarithms in \eqref{eq:Q-lower-Kn} and observing that
$c_\alpha$ is independent of $n$ gives
\begin{align}
\liminf_{n\to\infty}
\frac{1}{n}
\log
\mathbb E_{\mathcal C_n}
Q_\alpha
\left(
\rho_{E^n}^{\mathcal C_n}
\middle\|
\rho_E^{\otimes n}
\right)
\geq
-
\sup_{\alpha\leq a\leq1}
\frac{\alpha(1-a)}a
\left\{
\widetilde I_a^{
\frac{\alpha-a}{\alpha(1-a)}
}(X:E)_{\rho_{XE}}
-R
\right\}.
\label{eq:Q-exponent-lower}
\end{align}
Since
\begin{align*}
D_\alpha
\left(
\rho_{\mathcal C_nE^n}
\middle\|
\rho_{\mathcal C_n}\otimes\rho_E^{\otimes n}
\right)=
\frac1{\alpha-1}
\log
\mathbb E_{\mathcal C_n}
Q_\alpha
\left(
\rho_{E^n}^{\mathcal C_n}
\middle\|
\rho_E^{\otimes n}
\right),
\end{align*}
and $\alpha-1<0$, \eqref{eq:Q-exponent-lower} yields
\begin{align}
\limsup_{n\to\infty}
\frac{1}{n}
D_\alpha
\left(
\rho_{\mathcal C_nE^n}
\middle\|
\rho_{\mathcal C_n}\otimes\rho_E^{\otimes n}
\right)
\leq
\sup_{\alpha\leq a\leq1}
\frac{\alpha(1-a)}{a(1-\alpha)}
\left\{
\widetilde I_a^{
\frac{\alpha-a}{\alpha(1-a)}
}(X:E)_{\rho_{XE}}
-R
\right\}.
\label{eq:sc-upper-alpha-half-one}
\end{align}
Together with Theorem~\ref{Thm:channel-res-renyi}, this proves
\eqref{eq:exact-sc-alpha-half-one} for
$\alpha\in(\frac12,1)$.

It remains to consider $\alpha=\frac12$.  We prove the result by letting the R{\'e}nyi order
$\beta\downarrow\frac12$. For $\frac12\leq\beta\leq a<1$, set the
second parameter in the two-parameter club-sandwiched mutual information to
\[
\frac{\beta-a}{\beta(1-a)}.
\]
The two effective exponents appearing in its defining trace
functional are
\[
\frac{1-a}{2a}
\frac{\beta-a}{\beta(1-a)}
=
\frac{\beta-a}{2\beta a},
\qquad
\frac{1-a}{2a}
\left(
1-\frac{\beta-a}{\beta(1-a)}
\right)
=
\frac{1-\beta}{2\beta}.
\]
Moreover,
\[
\frac{1-2a}{1-a}
\leq
\frac{\beta-a}{\beta(1-a)}
\leq0,
\]
because
\[
\frac{\beta-a}{\beta(1-a)}
-
\frac{1-2a}{1-a}
=
\frac{a(2\beta-1)}{\beta(1-a)}
\geq0.
\]
Thus, such club-sandwiched mutual information satisfies the data processing inequality; see \cite{RGT2024quantum}. We first check that the map
\begin{equation}
(\beta,a)\longmapsto
\frac{\beta(1-a)}{a(1-\beta)}
\left\{
\widetilde I_a^{\frac{\beta-a}{\beta(1-a)}}
(X:E)_{\rho_{XE}}-R
\right\}
\label{eq:endpoint-objective-map}
\end{equation}
is jointly upper semicontinuous on
\[
\left\{
(\beta,a):
\frac12\leq\beta\leq\frac34,\
\beta\leq a\leq1
\right\},
\]
where the value at $a=1$ is defined to be zero. Indeed, for $a<1$, the definition of the club-sandwiched mutual
information gives
\begin{align}
&\frac{\beta(1-a)}{a(1-\beta)}
\left\{
\widetilde I_a^{\frac{\beta-a}{\beta(1-a)}}
(X:E)_{\rho_{XE}}-R
\right\}
\nonumber\\
&=
-\frac{\beta}{a(1-\beta)}
\log
\inf_{\sigma_E\in\mathcal D(E)}
\operatorname{Tr}\Big[
(\rho_X\otimes\sigma_E)^{\frac{\beta-a}{2\beta a}}
(\rho_X\otimes\rho_E)^{\frac{1-\beta}{2\beta}}
\rho_{XE}
\nonumber\\
&\hspace{42mm}\times
(\rho_X\otimes\rho_E)^{\frac{1-\beta}{2\beta}}
(\rho_X\otimes\sigma_E)^{\frac{\beta-a}{2\beta a}}
\Big]^a
-
\frac{\beta(1-a)}{a(1-\beta)}R.
\label{eq:endpoint-objective-trace}
\end{align}
Here we used the fact that $a-1<0$, so that the supremum in the
definition becomes an infimum of the underlying trace functional.
By \cite[Lemmas~A.1 and A.2]{RGT2024quantum}, the
extended-valued trace functional in
\eqref{eq:endpoint-objective-trace} satisfies
\begin{align}
&\operatorname{Tr}\Big[
(\rho_X\otimes\sigma_E)^{\frac{\beta-a}{2\beta a}}
(\rho_X\otimes\rho_E)^{\frac{1-\beta}{2\beta}}
\rho_{XE}
(\rho_X\otimes\rho_E)^{\frac{1-\beta}{2\beta}}
(\rho_X\otimes\sigma_E)^{\frac{\beta-a}{2\beta a}}
\Big]^a
\nonumber\\
&=
\sup_{\varepsilon>0}
\operatorname{Tr}\Big[
\bigl(\rho_X\otimes(\sigma_E+\varepsilon I_E)\bigr)^{
\frac{\beta-a}{2\beta a}}
(\rho_X\otimes\rho_E)^{\frac{1-\beta}{2\beta}}
\rho_{XE}
\nonumber\\
&\hspace{36mm}\cdot
(\rho_X\otimes\rho_E)^{\frac{1-\beta}{2\beta}}
\bigl(\rho_X\otimes(\sigma_E+\varepsilon I_E)\bigr)^{
\frac{\beta-a}{2\beta a}}
\Big]^a.
\label{eq:endpoint-regularization-short}
\end{align}
For every fixed $\varepsilon>0$, the expression on the right-hand
side is continuous in $(\beta,a,\sigma_E)$. Therefore, the left-hand
side is a pointwise supremum of
continuous functions and hence jointly lower semicontinuous. Since $\mathcal D(E)$ is compact, its infimum
over $\sigma_E$ is lower semicontinuous in $(\beta,a)$. Consequently,
the right-hand side of \eqref{eq:endpoint-objective-trace} is jointly
upper semicontinuous.

At $a=1$, the trace optimization reduces to
\begin{align}
\inf_{\sigma_E\in\mathcal D(E)}
\operatorname{Tr}
\rho_E^{1/\beta}
\sigma_E^{-(1-\beta)/\beta}
=1.
\label{eq:endpoint-a-one}
\end{align}
Indeed, this follows from the nonnegativity of the Petz R{\'e}nyi
divergence of order $1/\beta>1$, and equality is attained at
$\sigma_E=\rho_E$. Hence the expression in
\eqref{eq:endpoint-objective-map} has the upper-semicontinuous
extension
\[
\left.
\frac{\beta(1-a)}{a(1-\beta)}
\left\{
\widetilde I_a^{\frac{\beta-a}{\beta(1-a)}}(X:E)-R
\right\}
\right|_{a=1}
=0.
\]
Now let $\beta_k\downarrow\frac12$. By the joint upper
semicontinuity and the compactness of the
optimization intervals, we have
\begin{align}
&\limsup_{k\to\infty}
\sup_{\beta_k\leq a\leq1}
\frac{\beta_k(1-a)}{a(1-\beta_k)}
\left\{
\widetilde I_a^{\frac{\beta_k-a}
{\beta_k(1-a)}}(X:E)_{\rho_{XE}}-R
\right\}
\nonumber\\
&\hspace{30mm}\leq
\sup_{\frac12\leq a\leq1}
\frac{1-a}{a}
\left\{
\widetilde I_a^{\frac{1-2a}{1-a}}
(X:E)_{\rho_{XE}}-R
\right\}.
\label{eq:endpoint-sup-limit}
\end{align}
% To see this explicitly, one may choose a maximizing sequence
% $a_k\in[\beta_k,1]$, pass to a subsequence such that
% $a_k\to a_*\in[\frac12,1]$, and then apply the joint upper
% semicontinuity of \eqref{eq:endpoint-objective-map}.
Finally, the sandwiched R{\'e}nyi divergence is nondecreasing in its
order. Hence, for every $\beta\in(\frac12,1)$,
\begin{align*}
\limsup_{n\to\infty}
\frac1n
D_{\frac12}\left(
\rho_{\mathcal C_nE^n}
\middle\|
\rho_{\mathcal C_n}\otimes\rho_E^{\otimes n}
\right)
&\leq
\limsup_{n\to\infty}
\frac1n
D_{\beta}\left(
\rho_{\mathcal C_nE^n}
\middle\|
\rho_{\mathcal C_n}\otimes\rho_E^{\otimes n}
\right)
\\
&\leq
\sup_{\beta\leq a\leq1}
\frac{\beta(1-a)}{a(1-\beta)}
\left\{
\widetilde I_a^{\frac{\beta-a}{\beta(1-a)}}
(X:E)_{\rho_{XE}}-R
\right\}.
\end{align*}
Letting $\beta\downarrow\frac12$ and applying
\eqref{eq:endpoint-sup-limit} yields
\begin{align*}
&\limsup_{n\to\infty}
\frac1n
D_{\frac12}\left(
\rho_{\mathcal C_nE^n}
\middle\|
\rho_{\mathcal C_n}\otimes\rho_E^{\otimes n}
\right)\leq
\sup_{\frac12\leq a\leq1}
\frac{1-a}{a}
\left\{
\widetilde I_a^{\frac{1-2a}{1-a}}
(X:E)_{\rho_{XE}}-R
\right\},
\end{align*}
where the $a=1$ term is understood as its continuous extension,
whose value is zero.
\end{proof}

\section{Strong Converse Exponent of Quantum Soft Covering for $\alpha\in[1,\infty)$}
This section is devoted to the proof of Theorem~\ref{thm:sc-sc}. The first two subsections are for the one-shot achievability part, and in the last subsection we prove Theorem~\ref{thm:sc-sc}.

\subsection{One-Shot Achievability Part for $\alpha\in(1,2]$}
\begin{proposition}
\label{prop:strong-converse-achievability}
Let $\rho_{XE}=\sum_{x\in\mathcal{X}} P_X(x)\, |x\rangle\langle x| \otimes \rho_E^x$
be a C-Q state and $
\mathcal C=\{X(1),\ldots,X(M)\}$
be a random codebook whose codewords are independently drawn according
to $P_X$. For the codebook-induced
state $\rho_{E}^{\mathcal{{C}}}$ and the target marginal state \(\rho_E\),
we have for $\alpha\in(1,2],$
\begin{equation*}
\mathbb{E}_{\mathcal{C}}Q_{\alpha}\left(\rho_{E}^{\mathcal{C}}\|\rho_{E}\right)\leq\nu(\rho_{E})^{\alpha}\left(\exp\left\{ (\alpha-1)(D_{\alpha}(\rho_{XE}\|\rho_X\otimes\rho_{E})-\log M)\right\} +1\right).
\end{equation*}
\end{proposition}
\begin{proof}
Let the spectral projections of $\rho_{E}$ be $\{\Pi_{i}\}_{i\in\mathcal{I}}$,
and the corresponding eigenvalues be $\{\lambda_{i}\}_{i\in\mathcal{I}}$.
Then, with the pinching inequality, we can bound $\mathbb{E}_{\mathcal{C}}Q_{\alpha}\left(\rho_{E}^{\mathcal{C}}\|\rho_{E}\right)$
as follows.

\begin{align*}
& \mathbb{E}_{\mathcal{C}}Q_{\alpha}\left(\rho_{E}^{\mathcal{C}}\|\rho_{E}\right)\nb\\
\stackrel{(a)}{\leq} & \nu(\rho_{E})^{\alpha}\mathbb{E}_{\mathcal{C}}Q_{\alpha}\left(\mathcal{E}_{\rho_{E}}\left(\rho_{E}^{\mathcal{C}}\right)\|\rho_{E}\right)\nb\\
= & \nu(\rho_{E})^{\alpha}\mathbb{E}_{\mathcal{C}}\sum_{i=1}Q_{\alpha}\left(\sum_{m=1}\Pi_{i}\frac{1}{M}\rho_{E}^{X(m)}\Pi_{i}\|\lambda_{i}\Pi_{i}\right)\nb\\
= & \nu(\rho_{E})^{\alpha}\mathbb{E}_{\mathcal{C}}\sum_{i,m}\lambda_{i}^{1-\alpha}\frac{1}{M}{\rm Tr}\Pi_{i}\rho_{E}^{X(m)}\Pi_{i}\left(\frac{1}{M}\Pi_{i}\rho_{E}^{X(m)}\Pi_{i}+\sum_{m'\neq m}\frac{1}{M}\Pi_{i}\rho_{E}^{X(m')}\Pi_{i}\right)^{\alpha-1}\nb\\
\stackrel{(b)}{\leq} & \nu(\rho_{E})^{\alpha}\mathbb{E}_{\mathcal{C}}\sum_{i,m}\lambda_{i}^{1-\alpha}\frac{1}{M}{\rm Tr}\Pi_{i}\rho_{E}^{X(m)}\Pi_{i}\left(\frac{1}{M}\Pi_{i}\rho_{E}^{X(m)}\Pi_{i}+\mathbb{E}_{\mathcal{C}}\sum_{m'\neq m}\frac{1}{M}\Pi_{i}\rho_{E}^{X(m')}\Pi_{i}\right)^{\alpha-1}\nb\\
= & \nu(\rho_{E})^{\alpha}\mathbb{E}_{\mathcal{C}}\sum_{i,m}\lambda_{i}^{1-\alpha}\frac{1}{M}{\rm Tr}\Pi_{i}\rho_{E}^{X(m)}\Pi_{i}\left(\frac{1}{M}\Pi_{i}\rho_{E}^{X(m)}\Pi_{i}+\frac{M-1}{M}\Pi_{i}\rho_{E}\Pi_{i}\right)^{\alpha-1}\nb\\
\stackrel{(c)}{\leq} & \nu(\rho_{E})^{\alpha}\left(\mathbb{E}_{\mathcal{C}}\sum_{i,m}\lambda_{i}^{1-\alpha}{\rm Tr}\left(\frac{1}{M}\Pi_{i}\rho_{E}^{X(m)}\Pi_{i}\right)^{\alpha}+\left(\frac{M-1}{M}\right)^{\alpha-1}\right)\nb\\
= & \nu(\rho_{E})^{\alpha}\left(M^{1-\alpha}{\rm Tr}\sum_{x}P_{X}(x)\mathcal{E}_{\rho_{E}}(\rho_{E}^{x})^{\alpha}\rho_{E}^{1-\alpha}+\left(\frac{M-1}{M}\right)^{\alpha-1}\right)\nb\\
= & \nu(\rho_{E})^{\alpha}\left(\exp\left\{ (\alpha-1)\left(D_{\alpha}(\mathcal{E}_{\rho_{E}}(\rho_{XE})\|\rho_{X}\otimes\rho_{E})-\log M\right)\right\} +\left(\frac{M-1}{M}\right)^{\alpha-1}\right)\nb\\
\stackrel{(d)}{\leq} & \nu(\rho_{E})^{\alpha}\left(\exp\left\{ (\alpha-1)\left(D_{\alpha}(\rho_{XE}\|\rho_{X}\otimes\rho_{E})-\log M\right)\right\} +\left(\frac{M-1}{M}\right)^{\alpha-1}\right),
\end{align*}
where $(a)$ is by pinching inequality, $(b)$ follows from the matrix concavity
of $x\mapsto x^{\alpha-1}$, $(c)$ is since $\Pi_{i}\rho_{E}^{X(m)}\Pi_{i}$
commutes with $\lambda_{i}\Pi_{i}$,
$(d)$ is due to the data processing inequality of sandwiched R\'enyi
divergence.
\end{proof}

\subsection{One-Shot Achievability Part for $\alpha\in(2,\infty)$}
In \cite{liQiuZhang2026reliability}, the present authors have established the following one-shot achievability bound for
$\alpha\in(2,\infty)$.
\begin{lemma}\label{prop:two-sided-one-shot-detailed}
Let $\alpha>2$, $R\geq0$ and $\rho_{XE}$ be a C-Q state. Let $M=2^R$ and $\mathcal C=\{X_1,\cdots,X_M\}$ be a random codebook whose codewords are sampled independently according to $P_X$. There exists a constant $C_\alpha<\infty$, depending only on
$\alpha$, such that
\begin{align}
\mathbb E_{\mathcal C}
Q_{\alpha}\left(
\rho_E^{\mathcal C}\middle\|\rho_E
\right)\leq1+ C_\alpha\left(2^{-\gamma^\prime(\alpha)}+2^{-\gamma(\alpha-1)}
\right),
\label{eq:add-one-shot-upper}
\end{align}
where $\gamma(t):=t(R-I_{\alpha}(X:E)_{\rho_{XE}})$ and $\gamma'(t):=R-I_{2}^{(t)}(X:E)_{\rho_{XE}}$.
\end{lemma}

 We will use the above bound to prove the desired one-shot achievability bound of the strong converse exponent for quantum soft covering.
 
\begin{proposition}\label{prop:one-shot-upper-s>1}
Let $\alpha>2$, $R\geq0$ and $\rho_{XE}$ be a C-Q state. Let $M=2^R$ and $\mathcal C=\{X_1,\cdots,X_M\}$ be a random codebook whose codewords are sampled independently according to $P_X$. There exists a constant $C_\alpha<\infty$, depending only on
$\alpha$, such that
\begin{align}
\mathbb E_{\mathcal C}
Q_{\alpha}\left(
\rho_E^{\mathcal C}\middle\|\rho_E
\right)
&\leq
1+2C_\alpha\max\{1,\exp\left\{
(\alpha-1)\left(
I_{\alpha}(X:E)_{\rho_{XE}}-R
\right)
\right\}\}.
\label{eq:one-shot-upper-s>1-unpinched}
\end{align}
\end{proposition}
\begin{proof}
By Lemma~\ref{cor:additivity-ordering-information}, $
I_{2}^{(\alpha)}(X:E)_{\rho_{XE}}
\leq
I_\alpha(X:E)_{\rho_{XE}}.$
Therefore, if
$R\geq I_\alpha(X:E)_{\rho_{XE}}$, then
\begin{align}
(\alpha-1)(R-I_{\alpha}\left(X:E\right)_{\rho_{XE}})\geq0,\\
R-I_{2}^{(\alpha)}(X:E)_{\rho_{XE}}\geq0.
\end{align}
Applying Lemma~\ref{prop:two-sided-one-shot-detailed}, we have
\begin{equation}
\mathbb E_{\mathcal C}
Q_{\alpha}\left(
\rho_E^{\mathcal C}\middle\|\rho_E
\right)
\leq
1+2C_\alpha.\label{eq-sc-1}
\end{equation}
If $R\leq I_\alpha(X:E)_{\rho_{XE}}$, we obtain
\begin{align}
-(\alpha-1)(R-I_{\alpha}\left(X:E\right)_{\rho_{XE}})\geq -(R-I_{2}^{(\alpha)}(X:E)_{\rho_{XE}}).
\end{align}
Applying Lemma~\ref{prop:two-sided-one-shot-detailed} again, we get
\begin{equation}
\mathbb E_{\mathcal C}
Q_{\alpha}\left(
\rho_E^{\mathcal C}\middle\|\rho_E
\right)
\leq
1+2C_\alpha\exp\left\{
(\alpha-1)\left(
I_{\alpha}(X:E)_{\rho_{XE}}-R
\right)
\right\}.\label{eq-sc-2}
\end{equation}
Combining \eqref{eq-sc-1} and \eqref{eq-sc-2} yields the desired result.
\end{proof}

\begin{proof}[Proof of Theorem~\ref{thm:sc-sc}.]
For any $n\in\mathbb{N}$,
let $\mathcal{C}_{n}:=\{ X^{n}(1),...,X^{n}(2^{nR})\} $ be
a random i.i.d. codebook,
where each $X^{n}(m)$ is independently drawn from $P_{X}^{\otimes n}$.
On the one hand, for any $\alpha>1$, we have
\begin{align}
& \mathbb{E}_{\mathcal{C}_{n}}Q_{\alpha}(\rho_{E^{n}}^{\mathcal{C}_{n}}\|\rho_{E}^{\otimes n})\nb\\
= & \mathbb{E}_{\mathcal{C}_{n}}{\rm Tr}\left(\left(\rho_{E}^{\otimes n}\right)^{\frac{1-\alpha}{2\alpha}}\sum_{m}\frac{1}{2^{nR}}\rho_{E^{n}}^{X^n(m)}\left(\rho_{E}^{\otimes n}\right)^{\frac{1-\alpha}{2\alpha}}\right)^{\alpha}\nb\\
\geq & \mathbb{E}_{\mathcal{C}_{n}}{\rm Tr}\sum_{m}\left(\frac{1}{2^{nR}}\left(\rho_{E}^{\otimes n}\right)^{\frac{1-\alpha}{2\alpha}}\rho_{E^{n}}^{X^n(m)}\left(\rho_{E}^{\otimes n}\right)^{\frac{1-\alpha}{2\alpha}}\right)^{\alpha}\nb\\
= & |\mathcal{C}_{n}|^{1-\alpha}Q_{\alpha}\left(\rho_{XE}^{\otimes n}\|\rho_{X}^{\otimes n}\otimes\rho_{E}^{\otimes n}\right),
\end{align}
where the inequality is due to McCarthy’s inequality.
Therefore, for any $\alpha>1$, it holds that
\begin{align}
\varGamma_{\rm sc}^{(\alpha)}(\rho_{XE},R) =\liminf_{n\rightarrow\infty}\frac{1}{n(\alpha-1)}\log\mathbb{E}_{\mathcal{C}_{n}}Q_{\alpha}(\rho_{E^{n}}\|\rho_{E}^{\otimes n})
\geq|I_{\alpha}(X:E)_{\rho_{XE}}-R|^{+}.\label{sc-1}
\end{align}

On the other hand, for $\alpha\in(1,2]$, we have
\begin{align}
& \varGamma_{\rm sc}^{(\alpha)}(\rho_{XE},R)\nb\\
= & \liminf_{n\rightarrow\infty}\frac{1}{n(\alpha-1)}\log\mathbb{E}_{\mathcal{C}_{n}}Q_{\alpha}(\rho_{E^{n}}^{\mathcal{C}_{n}}\|\rho_{E}^{\otimes n})\nb\\
\leq & \liminf_{n\rightarrow\infty}\frac{1}{n(\alpha-1)}\log\left\{\nu(\rho_{E}^{\otimes n})^{\alpha}\left(\exp\left\{ (\alpha-1)\left(D_{\alpha}(\rho_{XE}^{\otimes n}\|\rho_{X}^{\otimes n}\otimes\rho_{E}^{\otimes n})-nR\right)\right\} +1\right)\right\}\nb\\
\leq & \liminf_{n\rightarrow\infty}\left\{ |I_{\alpha}(X:E)_{\rho_{XE}}-R|^{+}+\frac{\alpha}{n(\alpha-1)}\log\nu(\rho_{E}^{\otimes n})\right\} \nb\\
= & |I_{\alpha}(X:E)_{\rho_{XE}}-R|^{+},\label{sc-2}
\end{align}
where the first inequality is by Proposition \ref{prop:strong-converse-achievability}
and the last inequality follows from the fact that $\nu(\rho_{E}^{\otimes n})$
is a polynomial of $n$.  For $\alpha>2$, applying
Proposition~\ref{prop:one-shot-upper-s>1}, we obtain
\begin{align}
\varGamma_{\rm sc}^{(\alpha)}(\rho_{XE},R)
= \liminf_{n\rightarrow\infty}\frac{1}{n(\alpha-1)}\log\mathbb{E}_{\mathcal{C}_{n}}Q_{\alpha}(\rho_{E^{n}}\|\rho_{E}^{\otimes n})
\leq
|I_{\alpha}(X:E)_{\rho_{XE}}-R|^+.
\label{sc-3}
\end{align}
Combining \eqref{sc-1}-\eqref{sc-3} yields the desired result for $\alpha>1$.

Finally, we prove the desired result for the case $\alpha=1$. Let $J$ be uniformly distributed over
$\{1,\ldots,2^{nR}\}$ and independent of the random codebook
$\mathcal C_n$. Consider the state
\[
\omega_{\mathcal C_nJE^n}
:=
\sum_{c_n}P_{\mathcal C_n}(c_n)
|c_n\rangle\langle c_n|
\otimes
\frac{1}{2^{nR}}
\sum_{j=1}^{2^{nR}}
|j\rangle\langle j|
\otimes\rho_{E^n}^{x^n(j)}.
\]
Since $
\omega_{E^n}
=
\mathbb E_{\mathcal C_n}
\rho_{E^n}^{\mathcal C_n}
=
\rho_E^{\otimes n},$
we have
\[
D\left(
\rho_{\mathcal C_nE^n}
\,\middle\|\,
\rho_{\mathcal C_n}\otimes\rho_E^{\otimes n}
\right)
=
I(\mathcal C_n:E^n)_\omega.
\]
Let $Y^n:=X^n(J)$ be the codeword selected by $J$. Then $Y^n$ follows the distribution $P_X^{\otimes n}$, and the state of $E^n$ conditioned on $(\mathcal C_n,J)$ depends only on $Y^n$. Since $Y^n$ is a
deterministic function of $(\mathcal C_n,J)$, it follows that
\[
I(\mathcal C_nJ:E^n)_\omega
=
I(Y^n:E^n)_\omega
=
nI(X:E)_{\rho_{XE}}.
\]
On the other hand, the chain rule gives
\[
I(\mathcal C_nJ:E^n)_\omega
=
I(\mathcal C_n:E^n)_\omega
+
I(J:E^n\mid\mathcal C_n)_\omega.
\]
Since $J$ is classical and takes at most $2^{nR}$ values,
\[
I(J:E^n\mid\mathcal C_n)_\omega
\leq
H(J\mid\mathcal C_n)
=
nR.
\]
Consequently,
\[
D\left(
\rho_{\mathcal C_nE^n}
\,\middle\|\,
\rho_{\mathcal C_n}\otimes\rho_E^{\otimes n}
\right)
\geq
nI(X:E)_{\rho_{XE}}-nR.
\]
By the non-negativity of the relative entropy,
\[
D\left(
\rho_{\mathcal C_nE^n}
\,\middle\|\,
\rho_{\mathcal C_n}\otimes\rho_E^{\otimes n}
\right)
\geq
\left|nI(X:E)_{\rho_{XE}}-nR\right|^+.
\]
Dividing by $n$, taking the limit inferior, and using
$\frac{1}{n}nR\to R$, we obtain
\[
\varGamma_{\mathrm{sc}}^{(1)}(\rho_{XE},R)
\geq
\bigl|I(X:E)_{\rho_{XE}}-R\bigr|^+.
\]
By the monotonicity of the sandwiched R{\'e}nyi divergence in its
order, together with the continuity of $x\mapsto |x|^+$ and $
\lim_{\alpha\searrow1} I_{\alpha}(X:E)_{\rho_{XE}}
= I(X:E)_{\rho_{XE}},$
we obtain
\begin{align}
\varGamma_{\mathrm{sc}}^{(1)}(\rho_{XE},R)
&\leq \limsup_{\alpha\searrow 1}
\varGamma_{\mathrm{sc}}^{(\alpha)}(\rho_{XE},R) \notag\\
&\leq \lim_{\alpha\searrow 1}
\bigl|I_{\alpha}(X:E)_{\rho_{XE}}-R\bigr|^+ \notag\\
&=\bigl|I(X:E)_{\rho_{XE}}-R\bigr|^+.
\end{align}
\end{proof}

\section{Strong Converse Exponent of Quantum Privacy Amplification}
To prove the achievability part of
Theorem~\ref{thm:5-15}, we employ the random binning function.
A random function
\[
F:\mathcal X\to\mathcal Z=\{1,2,\ldots,M\}
\]
is called the random binning function if the family of random
variables $\{F(x)\}_{x\in\mathcal X}$ is mutually independent and
each $F(x)$ is uniformly distributed over $\mathcal Z$. That is,
for every $x\in\mathcal X$ and $z\in\mathcal Z$,
\begin{equation*}
\Pr\{F(x)=z\}
=
\frac{1}{|\mathcal Z|}
=
\frac{1}{M}.
\end{equation*}
Given $\rho_{XE}=\sum_{x\in\mathcal{X}} P_X(x)\, |x\rangle\langle x| \otimes \rho_E^x$, the C-Q state induced by the random binning function is defined as
\begin{equation*}
\mathcal{R}_{F}(\rho_{XE})
=\sum_{x\in\mathcal{X}} P_X(x)
|F(x)\rangle\langle F(x)| \otimes \rho_{E}^{x}.
\end{equation*}
Note that $\mathcal{R}_{F}(\rho_{XE})$ is a random state due to the randomness of $F$.

In \cite{liQiuZhang2026reliability}, the present authors have established the following one-shot achievability bound for
$\alpha\in(2,\infty)$.

\begin{lemma}\label{thm:one-shot}
Let $\alpha\in(2,\infty)$, $M=2^R\geq2$ and $F:\mathcal{X}\to\mathcal{Z}=\{1,2,\cdots,M\}$ be the random binning function.
There exists a constant $C_\alpha<\infty$, depending only on
$\alpha$, such that
\begin{align}
\mathbb E_FQ_\alpha\left(
\mathcal R_F(\rho_{XE})
\middle\|
\frac{\mathbbm{1}_\mathcal{Z}}{|\mathcal{Z}|}\otimes\rho_E
\right)\leq 1+
C_\alpha\Big[
2^{-\eta'(\alpha)}
+2^{-\frac{1}{2}\eta'(\alpha)
-\frac{1}{2}\eta(\alpha-1)}
+2^{1-\eta(\alpha-1)}
\Big],
\label{eq:PA-upper}
\end{align}
where $\eta(t):=t(H_{1+t}\left(X|E\right)_{\rho_{XE}}-R)$, $\eta'(t):=H_{2}^{(t)}(X|E)_{\rho_{XE}}-R$.
\end{lemma}

 We will use the above bound to prove the desired one-shot achievability bound of the strong converse exponent for quantum privacy amplification.
 
\begin{proposition}\label{prop:pa-one-shot-upper-s>1}
Let $\alpha\in(2,\infty)$, $M=2^R\geq2$ and $F:\mathcal{X}\to\mathcal{Z}=\{1,2,\cdots,M\}$ be the random binning function.
There exists a constant $C_\alpha<\infty$, depending only on
$\alpha$, such that
\begin{align}
&\mathbb E_F
Q_{\alpha}\left(
\mathcal R_F(\rho_{XE})
\Big\|
\frac{\mathbbm 1_{\mathcal Z}}{M}\otimes\rho_E
\right)
\leq
1+4C_\alpha
\max\{1,\exp\left\{
(\alpha-1)\left(
R-H_{\alpha}(X|E)_{\rho_{XE}}
\right)
\right\}\}.
\label{eq:pa-one-shot-upper-s>1}
\end{align}
\end{proposition}

\begin{proof}
By Lemma~\ref{cor:additivity-ordering-information}, $H_{2}^{(\alpha)}(X|E)_{\rho_{XE}}\geq H_\alpha(X|E)_{\rho_{XE}}$. If $R\leq H_\alpha(X|E)_{\rho_{XE}}$, we obtain
\begin{align}
(\alpha-1)(H_{\alpha}\left(X|E\right)_{\rho_{XE}}-R)\geq0,\\
H_{2}^{(\alpha)}(X|E)_{\rho_{XE}}-R\geq0.
\end{align}
Applying Lemma~\ref{thm:one-shot}, we have
\begin{equation}
\mathbb E_F
Q_{\alpha}\left(
\mathcal R_F(\rho_{XE})
\Big\|
\frac{\mathbbm 1_{\mathcal Z}}{M}\otimes\rho_E
\right)
\leq
1+4C_\alpha.\label{eq-pa-1}
\end{equation}
If $R\geq H_\alpha(X|E)_{\rho_{XE}}$, we obtain
\begin{align}
-(\alpha-1)(H_{\alpha}\left(X|E\right)_{\rho_{XE}}-R)\geq -(H_{2}^{(\alpha)}(X|E)_{\rho_{XE}}-R).
\end{align}
Applying Lemma~\ref{thm:one-shot} again, we have
\begin{equation}
\mathbb E_F
Q_{\alpha}\left(
\mathcal R_F(\rho_{XE})
\Big\|
\frac{\mathbbm 1_{\mathcal Z}}{M}\otimes\rho_E
\right)
\leq
1+4C_\alpha\exp\left\{
(\alpha-1)\left(
R-H_{\alpha}(X|E)_{\rho_{XE}}
\right)
\right\}.\label{eq-pa-2}
\end{equation}
Combining \eqref{eq-pa-1} and \eqref{eq-pa-2} yields the desired result.
\end{proof}

\begin{proof}[Proof of Theorem~\ref{thm:5-15}.]
In \cite{LYH2023tight}, for any $\alpha>2$, it holds that
\begin{equation*}
\varGamma_{\rm pa}^{(\alpha)}(\rho_{XE},R)\geq \left|R-H_{\alpha}(X|E)_{\rho_{XE}}
\right|^+.
\end{equation*}
We now apply Proposition~\ref{prop:pa-one-shot-upper-s>1} to prove the desired upper bound.  Let
$F_n:\mathcal X^n\to\mathcal Z_n=\{1,\cdots,2^{nR}\}$ be a random binning function.
The additivity of the sandwiched R{\'e}nyi conditional
entropy gives
\begin{equation}
H_{\alpha}(X^n|E^n)_{\rho^{\otimes n}}
=
nH_{\alpha}(X|E)_{\rho_{XE}}.
\label{eq:pa-conditional-entropy-additivity}
\end{equation}
Applying Proposition~\ref{prop:pa-one-shot-upper-s>1} with
$F\leftarrow F_n$ and
$\rho_{XE}\leftarrow\rho_{XE}^{\otimes n}$,
for any $\alpha>2$, we obtain
\begin{align}
\limsup_{n\to\infty}
\frac{1}{n}
\log
\mathbb E_{F_n}
Q_{\alpha}\left(
\mathcal R_{F_n}(\rho_{XE}^{\otimes n})
\Big\|
\frac{\mathbbm 1_{\mathcal Z_n}}{|\mathcal Z_n|}
\otimes\rho_E^{\otimes n}
\right)
\leq
\max\left\{
0,\,
(\alpha-1)\left(
R-H_{\alpha}(X|E)_{\rho_{XE}}
\right)\right\}.
\label{eq:pa-strong-converse-max}
\end{align}
Then, there exists a realization $f_n$ of $F_n$ such that
\begin{align*}
\limsup_{n\to\infty}
\frac{1}{n}
D_{\alpha}\left(
\mathcal R_{f_n}(\rho_{XE}^{\otimes n})
\Big\|
\frac{\mathbbm 1_{\mathcal Z_n}}{|\mathcal Z_n|}
\otimes\rho_E^{\otimes n}
\right)
\leq
\left|
R-H_{\alpha}(X|E)_{\rho_{XE}}
\right|^+.
\end{align*}
From the definition of the strong converse exponent of privacy amplification, we have
\begin{equation*}
\varGamma_{\rm pa}^{(\alpha)}(\rho_{XE},R)\leq \left|R-H_{\alpha}(X|E)_{\rho_{XE}}
\right|^+.
\end{equation*}
\end{proof}

\section{Conclusion}

In this work, we have determined the exact strong converse exponent of
quantum soft covering under the sandwiched R{\'e}nyi divergence for all
orders $\alpha\in[\frac{1}{2},\infty)$. The resulting characterization
exhibits a transition at order one. For
$\alpha\in[\frac{1}{2},1)$, the exponent is governed by the
two-parameter club-sandwiched mutual information, whereas for
$\alpha\in[1,\infty)$, it is governed by the order $\alpha$
sandwiched R{\'e}nyi mutual information. Besides completing the strong
converse analysis of quantum soft covering, this result provides a
precise operational interpretation of the club-sandwiched mutual
information in the quantum setting.

We have also characterized the exact strong converse exponent of
privacy amplification against quantum side information for
$\alpha\in(2,\infty)$. Under the fixed marginal criterion, the exponent
is determined by the order $\alpha$ sandwiched R{\'e}nyi conditional
entropy. Together, these results reveal a complementary structure
between the two tasks: the failure of soft covering is governed by the
amount by which the relevant mutual information exceeds the coding
rate, whereas the failure of privacy amplification is governed by the
amount by which the extraction rate exceeds the corresponding
conditional entropy.

Together with earlier results
\cite{LiYao2024operational,rubboliTomamichel2026composable}, the present work completes the
characterization of the strong converse exponents associated with the
sandwiched R{\'e}nyi divergence over the full range of orders
$\alpha\in[\frac{1}{2},\infty)$. A natural next step is to determine
the exact strong converse exponents of quantum soft covering and
quantum privacy amplification under the trace distance. For quantum
soft covering, Cheng and Gao~\cite{chengGao2024softCovering} established a lower
bound on the strong converse exponent. For quantum privacy
amplification, lower bounds were obtained by Shen, Gao, and
Cheng~\cite{shenGaoCheng2023strongConversePA} and by Salzmann and
Datta~\cite{SalzmannDatta2022total}. These bounds provide only one side of
the corresponding exponents, and exact characterizations under the
trace distance remain open. Another interesting direction is to
investigate whether the two-parameter club-sandwiched mutual information and related
R{\'e}nyi information quantities admit further operational
interpretations in other quantum information processing tasks.

\appendix
\section{Auxiliary Lemmas}
\begin{lemma}[Reverse matrix Rosenthal inequality]
\label{lem:reverse-matrix-rosenthal}
Let $1<p<2$, and let
$B_1,\ldots,B_M$ be independent matrix-valued random variables
satisfying
\begin{equation*}
\mathbb E B_m=0,
\qquad
\mathbb E \Tr |B_m|^p<\infty,
\qquad
1\leq m\leq M.
\end{equation*}
Then there exists a constant $c_p>0$, depending only on $p$, such
that
\begin{align*}
&\left(
\mathbb E
\left\|
\sum_{m=1}^M e_{m1}\otimes B_m
\right\|_p^p
\right)^{1/p}\quad\geq
c_p
\inf_{B_m=D_m+C_m+R_m}
\Bigg\{
\left(
\sum_{m=1}^M
\mathbb E\operatorname{Tr}|D_m|^p
\right)^{1/p}
\\
&\hspace{35mm}
+
\left[
\operatorname{Tr}
\left(
\sum_{m=1}^M
\mathbb E C_m^*C_m
\right)^{p/2}
\right]^{1/p}
\\
&\hspace{35mm}
+
\left[
\sum_{m=1}^M
\operatorname{Tr}
\left(
\mathbb E R_mR_m^*
\right)^{p/2}
\right]^{1/p}
\Bigg\}.
\label{eq:reverse-matrix-rosenthal}
\end{align*}
The infimum is taken over all decompositions
$
B_m=D_m+C_m+R_m
$.
% where
% \begin{equation*}
% \mathbb E D_m
% =
% \mathbb E C_m
% =
% \mathbb E R_m
% =
% 0.
% \end{equation*}
\end{lemma}

\begin{proof}
Let $B_1,\cdots,B_M$ be matrix-valued random variables defined on $\mathcal{X}$ and take values in $\mathcal{B}(\mathcal{H})$, with probability distributions $\mu_1,\cdots,\mu_m$. Consider the finite von Neumann algebra
\[
\mathcal{M}
=
L_\infty(\mathcal{X}^m,\mu_1\times\cdots\times \mu_m)
\overline{\otimes}
\mathcal B(\mathcal H).
\]
equipped with the trace
\[
\tau(\cdot)=\mathbb E_{X_1,\cdots,X_m}\operatorname{Tr}(\cdot).
\]
We identify
$
\mathcal N=\mathcal B(\mathcal H)
$
with the subalgebra of $\mathcal M$. Then $\mathcal{E}=\mathbb E_{X_1,\cdots,X_m}$ is the trace preserving conditional expectation from $\mathcal{M}$ to $\mathcal{N}$.
Since $B_1,\cdots,B_m$ are independent, they are independent with respect to $\mathcal{E}$.

We now amplify $\mathcal M$ by the matrix algebra
$\mathcal B(\ell_2^M)$ and set
\[
\widehat{\mathcal M}
=
\mathcal M
\overline{\otimes}
\mathcal B(\ell_2^M),
\qquad
\widehat{\tau}
=
\tau\otimes\operatorname{Tr}_{\ell_2^M}.
\]
Applying the $1<p<2$ part of the matrix Rosenthal inequality
\cite[Theorem~3.3]{JungeXu2008} to the $M\times M$ matrix $(x_{mn})$ where
\[
x_{m1}=B_m,\quad x_{mn}=0,
\qquad
1\leq m\leq M,\quad 2\leq n\leq M,
\]
gives a constant $c_p>0$, depending only on $p$, such that
\begin{align}
\bigg\|
&\sum_{m,n=1}^M x_{mn}\otimes  e_{mn}
\bigg\|_{L_p(\widehat{\mathcal M})}\geq c_p\inf_{x_{mn}=D_{mn}+C_{mn}+R_{mn}}
\Bigg\{
\left(
\sum_{m,n=1}^M
\|D_{mn}\|_{L_p(\mathcal M)}^p
\right)^{\frac{1}{p}}
\nonumber\\
&
+\left(
\sum_{n=1}^{M}\left\|
\left(
\sum_{m=1}^M
\mathcal E(C_{mn}^*C_{mn})
\right)^{\frac{1}{2}}
\right\|_{L_p(\mathcal N)}\right)^{\frac{1}{p}}
\!+\!
\left(
\sum_{m=1}^M
\left\|\sum_{n=1}^M \left(
\mathcal E(R_{mn}R_{mn}^*)\right)^{\frac{1}{2}}
\right\|_{L_p(\mathcal N)}^p
\right)^{\frac{1}{p}}
\Bigg\}.
\label{eq:JX-one-column}
\end{align}
Here the infimum is taken over all decompositions such that
\[
\mathcal E(D_{mn})
=
\mathcal E(C_{mn})
=
\mathcal E(R_{mn})
=
0.
\]

By the definition of $(x_{mn})$, we have
$$ \bigg\|\sum_{m,n=1}^M x_{mn}\otimes  e_{mn}
\bigg\|_{L_p(\widehat{\mathcal M})}=\bigg\|\sum_{m,n=1}^M B_m\otimes  e_{m1}
\bigg\|_{L_p(\widehat{\mathcal M})}
=
\left(\mathbb E
\left\|
\sum_{m=1}^M B_m\otimes e_{m1}
\right\|_p^p\right)^{\frac{1}{p}}.$$

For the right-hand side in~\eqref{eq:JX-one-column}, when $n\geq 2$, we have $x_{mn}=0$. It then achieves its infimum with the decomposition $x_{mn}=0+0+0$ for $n\geq 2$. Hence, one only needs to optimize over all decompositions $x_{m1}=B_m=D_m+C_m+R_m$. Here we remove the centered decompositions restriction, for the infimum over a possibly larger class will not change the inequality. The right-hand side in~\eqref{eq:JX-one-column} then equals 
\begin{align*}
\inf_{B_m=D_m+C_m+R_m}\Bigg\{\left(\sum_{m=1}^M
\mathbb E\operatorname{Tr}|D_m|^p\right)^{\frac{1}{p}}
+
\Bigg(\operatorname{Tr}
\Bigg(&
\sum_{m=1}^M
\mathbb E C_m^*C_m
\Bigg)^{\frac{p}{2}}\Bigg)^{\frac{1}{p}} \nb\\
+&
\left(\sum_{m=1}^M
\operatorname{Tr}
\left(
\mathbb E R_mR_m^*
\right)^{\frac{p}{2}}\right)^{\frac{1}{p}}\Bigg\}.
\end{align*}
This completes the proof.
\end{proof}

\begin{theorem}[Complex interpolation]\label{complex-interpolation}
Let $(X_0,X_1)$ and $(Y_0,Y_1)$ be compatible Banach couples. Let
\[
T:X_0+X_1\longrightarrow Y_0+Y_1
\]
be a linear operator such that
\[
\|T\|_{X_0\to Y_0}\leq M_0, \quad \|T\|_{X_1\to Y_1}\leq M_1.
\]
Then, for every $0<\theta<1$, the operator $T$ induces a bounded
linear operator
\[
T:[X_0,X_1]_\theta\longrightarrow [Y_0,Y_1]_\theta,
\]
and
\[
\|T\|_{[X_0,X_1]_\theta\to [Y_0,Y_1]_\theta}
\leq M_0^{\,1-\theta}M_1^{\,\theta}.
\]
\end{theorem}

\begin{lemma}[McCarthy’s inequality~\cite{McCarthy1967Cp}]\label{Mc}
Let $P$ and $Q$ be positive semi-definite matrices. For $\alpha\in[0,1]$,
we have
\begin{equation*}
{\rm Tr}(P+Q)^{\alpha}\le{\rm Tr}P^{\alpha}+{\rm Tr}Q^{\alpha}.
\end{equation*}
For $\alpha>1$, we have
\begin{equation*}
{\rm Tr}(P+Q)^{\alpha}\ge{\rm Tr}P^{\alpha}+{\rm Tr}Q^{\alpha}.
\end{equation*}
\end{lemma}

\begin{lemma}[\cite{liQiuZhang2026reliability}]\label{cor:additivity-ordering-information}
Let $\rho_{XE}$ be a C-Q state. For any $\alpha\in(2,\infty)$, we have
\begin{align*}
H_\alpha(X|E)_{\rho_{XE}}
\leq&H_2^{(\alpha)}(X|E)_{\rho_{XE}}, \\
I_2^{(\alpha)}(X:E)_{\rho_{XE}}
\leq&I_\alpha(X:E)_{\rho_{XE}}.
\end{align*}
\end{lemma}

\section{Amalgamated $L_p$ Spaces}\label{amag}
We refer to \cite{JungeParcet2010mixed} for the details about the following discussion on the amalgamated $L_p$ spaces. Let $\mathcal M$ be a von Neumann algebra equipped with a normal faithful trace, let $\mathcal N\subseteq\mathcal M$ be a von Neumann
subalgebra. For $p\geq 1$, set
\[
\mathcal K
:=
\left\{
\left(\frac1u,\frac1v,\frac1q\right):
2\leq u,v\leq\infty,\quad
1\leq q\leq\infty,\quad
\frac1u+\frac1q+\frac1v=\frac{1}{p}\leq 1
\right\}.
\]
For $(u,v,q)\in \mathcal{K}$, we define
$$\|x\|_{u,q,v}=\inf \{\|a\|_{L_u(\mathcal N)}\|y\|_{L_q(\mathcal M)}\|b\|_{L_v(\mathcal N)}\,:\,x=ayb\in L_p(\mathcal M)\}$$
and
\[
L_u(\mathcal N)L_q(\mathcal M)L_v(\mathcal N)=\{x\in L_p(\mathcal{M})\,:\, \|x\|_{u,q,v}<\infty\}.
\]
Let $\mathcal E$ be a conditional expectation from $\mathcal M$ to $\mathcal N$. For $p\geq 1$, we have
$$ L_\infty(\mathcal N)L_p(\mathcal M)L_\infty(\mathcal N)=L_p(\mathcal{M}),\quad L_\infty(\mathcal N)L_2(\mathcal M)L_s(\mathcal N)=L_p^c(\mathcal M,\mathcal E),$$
where $1/p=1/2+1/s$.
For $j=0,1$, suppose that
\[
\left(\frac1{u_j},\frac1{v_j},\frac1{q_j}\right)
\in\mathcal K.
\]
Given $0\leq\theta\leq1$, define
\[
\frac1{u_\theta}
=
\frac{1-\theta}{u_0}+\frac{\theta}{u_1},
\qquad
\frac1{q_\theta}
=
\frac{1-\theta}{q_0}+\frac{\theta}{q_1},
\qquad
\frac1{v_\theta}
=
\frac{1-\theta}{v_0}+\frac{\theta}{v_1}.
\]
The following interpolation result is \cite[Theorem~3.2]{JungeParcet2010mixed}.
\begin{theorem}[Junge--Parcet]\label{amag-inter}
For every
$
\left(\frac1u,\frac1v,\frac1q\right)\in\mathcal K
$, the amalgamated space
\[
L_u(\mathcal N)L_q(\mathcal M)L_v(\mathcal N)
\]
is a Banach space. Moreover, one has the isometric complex
interpolation identity
\[
\left[
L_{u_0}(\mathcal N)L_{q_0}(\mathcal M)L_{v_0}(\mathcal N),
L_{u_1}(\mathcal N)L_{q_1}(\mathcal M)L_{v_1}(\mathcal N)
\right]_\theta
=
L_{u_\theta}(\mathcal N)
L_{q_\theta}(\mathcal M)
L_{v_\theta}(\mathcal N).
\]
\end{theorem}

\backmatter
\bmhead{AI Statement} The authors acknowledge the use of AI tools, including ChatGPT, for language
polishing, LaTeX editing, and exploratory mathematical discussions during the development and
preparation of this manuscript. Some ideas used in the proof-development process arose during
interactions with GPT-5.6 Sol. Suggestions from these AI interactions were subsequently examined, reformulated, incorporated into the manuscript, and independently verified by
the authors. The authors take full responsibility for all mathematical content in the final manuscript.

\bmhead{Acknowledgements}
The research of Shi-Bing Li was supported by NSFC under Grant 62571166. Hongsen Qiu and Xinyu Zhang were supported
by the National Natural Science Foundation of China (Grant No. 12371138 and No. W2441002).

\bmhead{Data Availability} Data sharing is not applicable to this article, since no data sets were used.

\bmhead{Conflict of Interest} The authors declare that they have no conflict of interest.

\bibliography{references}
\end{document}